\documentclass[a4paper,11pt]{article}
\usepackage{graphicx,epsfig}
\usepackage{fancyhdr,fancybox}
\usepackage{indentfirst}
\usepackage{verbatim}
\usepackage[sort&compress, numbers]{natbib}
\usepackage{geometry}
\usepackage{extarrows,chemarrow,xypic} 
\usepackage[small]{caption2}
\usepackage{color}
\usepackage{microtype}
\DisableLigatures[f]{encoding = *, family = *}

\usepackage{times}     

\usepackage{latexsym}
\usepackage{amsmath}   
\usepackage{amssymb}   
\usepackage{amsbsy}
\usepackage{amsthm}
\usepackage{amsfonts}
\usepackage{mathrsfs}  
\usepackage{bm}        
\usepackage{relsize}   
\usepackage{hyperref}  

\newcommand{\paperfont}{\fontsize{11pt}{1.2\baselineskip}\selectfont}
\begin{document}

\theoremstyle{definition}
\makeatletter
\thm@headfont{\bf}
\makeatother
\newtheorem{theorem}{Theorem}[section]
\newtheorem{definition}[theorem]{Definition}
\newtheorem{lemma}[theorem]{Lemma}
\newtheorem{proposition}[theorem]{Proposition}
\newtheorem{corollary}[theorem]{Corollary}
\newtheorem{remark}[theorem]{Remark}
\newtheorem{example}[theorem]{Example}

\lhead{}
\rhead{}
\lfoot{}
\rfoot{}

\renewcommand{\refname}{References}
\renewcommand{\figurename}{Fig.}
\renewcommand{\tablename}{Table}
\renewcommand{\proofname}{Proof}

\newcommand{\dnumiag}{\mathrm{diag}}
\newcommand{\tr}{\mathrm{tr}}
\newcommand{\dnum}{\mathrm{d}}
\newcommand{\Pnum}{\mathbb{P}}
\newcommand{\Enum}{\mathbb{E}}
\newcommand{\Rnum}{\mathbb{R}}
\newcommand{\Cnum}{\mathbb{C}}
\newcommand{\Znum}{\mathbb{Z}}
\newcommand{\Nnum}{\mathbb{N}}
\newcommand{\abs}[1]{\left\vert#1\right\vert}
\newcommand{\set}[1]{\left\{#1\right\}}
\newcommand{\norm}[1]{\left\Vert#1\right\Vert}
\newcommand{\innp}[1]{\langle {#1}\rangle}
\newcommand{\style}{\setlength{\itemsep}{1pt}\setlength{\parsep}{1pt}\setlength{\parskip}{1pt}}
\newcommand{\sech}{\mathrm{\,sech\,}}
\newcommand{\arcsec}{\mathrm{\,arcsec\,}}
\newcommand{\Ai}{\mathrm{Ai}}
\newcommand{\Bi}{\mathrm{Bi}}

\title{\textbf{Two-parameter localization and related phase transition for a Schr\"{o}dinger operator in balls and spherical shells}}
\author{Chen Jia$^{1}$,\;\;\;Zhimin Zhang$^{1,2}$,\;\;\;Lewei Zhao$^{1,*}$ \\
\footnotesize $^1$ Department of Mathematics, Wayne State University, Detroit, Michigan 48202, U.S.A.\\
\footnotesize $^2$ Beijing Computational Science Research Center, Beijing 100193, China. \\
\footnotesize $^*$ Correspondence: zhao.lewei@wayne.edu}
\date{}                              
\maketitle                           
\thispagestyle{empty}                

\paperfont

\begin{abstract}
Here we investigate the two-parameter high-frequency localization for the eigenfunctions of a Schr\"{o}dinger operator with a singular inverse square potential in high-dimensional balls and spherical shells as the azimuthal quantum number $l$ and the principal quantum number $k$ tend to infinity simultaneously, while keeping their ratio as a constant, generalizing the classical one-parameter localization for Laplacian eigenfunctions [SIAM J. Appl. Math. 73:780-803, 2013]. We prove that the eigenfunctions in balls are localized around an intermediate sphere whose radius is increasing with respect to the $l$-$k$ ratio. The eigenfunctions decay exponentially inside the localized sphere and decay polynomially outside. Furthermore, we discover a novel second-order phase transition for the eigenfunctions in spherical shells as the $l$-$k$ ratio crosses a critical value. In the supercritical case, the eigenfunctions are localized around a sphere between the inner and outer boundaries of the spherical shell. In the critical case, the eigenfunctions are localized around the inner boundary. In the subcritical case, no localization could be observed, giving rise to localization breaking. \\


\end{abstract}

\section{Introduction}
The eigenfunctions of an elliptic operator, especially the Laplace operator, and their localization phenomena have been extensively investigated in a wide range of mathematics and physics disciplines, including but not limited to spectral theory, probability theory, dynamical systems, acoustics, optics, wave mechanics, quantum mechanics, and condensed matter physics \cite{grebenkov2013geometrical}. An historical documentation of localization behavior dates back to the early 19th century, when Lord Rayleigh reported an acoustical phenomenon that occurred in the whispering gallery of St Paul's Cathedral in London \cite{rayleigh1910problem}: the whisper of one person propagated along the curved wall so that others can hear it from the opposite side of the gallery. Other historical examples of whispering galleries include the Gol Gumbaz mausoleum in Bijapur and the Echo Wall of the Temple of Heaven in Beijing. This acoustical effect and many related wave phenomena can be mathematically explained by the so-called whispering gallery modes, a certain type of Laplacian eigenmodes in a bounded domain, that are mostly distributed around the boundary of the domain and almost zero inside. The existence of whispering gallery modes in the limit of large eigenvalues has been constructed for an arbitrary two-dimensional domain with a smooth convex boundary by Keller and Rubinow \cite{keller1960asymptotic}.

In addition, Chen and coworkers have reported another type of localization phenomenon resulting from the so-called focusing modes \cite{chen1994visualization}. Under a different limit of large eigenvalues, they found that the Laplacian eigenmodes are localized in a small subdomain and are almost zero outside. Both the whispering gallery modes and focusing modes occur as the associated eigenvalue increases and thus are called high-frequency or high-energy eigenmodes. These and other localized eigenmodes have been intensively investigated for various domains, known as quantum billiards \cite{stockmann2000quantum, gutzwiller2013chaos}. Recently, a rigorous mathematical theory of high-frequency localization for Laplacian eigenfunctions in circular, spherical, and elliptical domains has been established by Nguyen and Grebenkov \cite{nguyen2013localization}. Readers may refer to \cite{jakobson2001geometric, grebenkov2013geometrical} for comprehensive reviews about the geometric properties of Laplacian eigenfunctions.

In quantum physics and condensed matter physics, Schr\"{o}dinger operators and their eigenfunctions play a fundamental role. One of the most important localization phenomena in condensed matter physics is Anderson localization \cite{anderson1958absence}, which describes the absence of diffusion of waves in disordered media and explains the metal-insulator transitions in semiconductors. Although it has been widely studied over the past fifty years \cite{lagendijk2009fifty, abrahams2010years}, there have been some recent developments regarding the mathematical analysis and computation on this subject \cite{filoche2012universal, arnold2016effective, arnold2019localization, arnold2019computing}. In addition, Schr\"{o}dinger operators with a singular inverse square potential have attracted increasing attention owing to its significant role in both mathematics and in physics. Mathematically, the inverse square potential has the same differential order as the Laplace operator \cite{li2017efficient}, while it usually invokes strong singularities of the Schr\"{o}dinger eigenfunctions and thus cannot be treated as a lower-order perturbation term \cite{kalf1975spectral, cao2006solutions, felli2006elliptic, felli2007schrodinger}. Physically, the inverse square potential serves as an intermediate threshold between regular potentials and singular potentials in nonrelativistic quantum mechanics \cite{case1950singular}. Furthermore, it also arises in many other scientific fields such as nuclear physics, molecular physics, and quantum cosmology \cite{frank1971singular}.

In this paper, we consider the eigenfunctions of the Schr\"{o}dinger operator
\begin{equation*}
Lu = -\triangle u+\frac{c^2}{|x|^2}u
\end{equation*}
with a singular inverse square potential, where $c\geq 0$ is the strength of the potential. Here we focus on the case of $d\geq 2$. Specifically, we consider the following eigenvalue problem with the Dirichlet boundary condition:
\begin{equation}\label{eigenproblem}\left\{
\begin{split}
& Lu = \lambda u\quad \textrm{in}\;\;\Omega, \\
& u = 0  \quad \textrm{on}\;\;\partial\Omega,
\end{split}\right.
\end{equation}
where $\Omega$ is a bounded region in $\Rnum^d$ with a smooth boundary. In fact, other boundary conditions such as the Neumann and Robin boundary conditions can be analyzed in a similar way.

Define the Sobolev spaces
\begin{equation*}
W^1(\Omega) = H^1(\Omega)\cap L^2_{r^{-2}}(\Omega),\;\;\;
W^1_0(\Omega) = H^1_0(\Omega)\cap L^2_{r^{-2}}(\Omega)
\end{equation*}
equipped with the norm
\begin{equation*}
\|u\|_{W^1(\Omega)} = (\|\nabla u\|^2+\|u\|^2_{r^{-2}})^{1/2}.
\end{equation*}
Then the variational form of \eqref{eigenproblem} is to find $\lambda\in\Rnum$ and $u\in W^1_0(\Omega)\setminus\{0\}$ such that
\begin{equation*}
(\nabla u,\nabla v)_{\Omega}+c^2(u,v)_{r^{-2},\Omega} = \lambda(u,v)_\Omega,\;\;\;v\in W^1_0(\Omega).
\end{equation*}
By the Sturm-Liouville theory, all eigenvalues of the eigenvalue problem \eqref{eigenproblem} can be listed as
\begin{equation*}
0 < \lambda_1 < \lambda_2 \leq\cdots \leq \lambda_n \rightarrow \infty.
\end{equation*}
Moreover, it is known that that $\lambda_n = O(n^{2/d})$ for any fixed $c\geq 0$ \cite{li2017efficient}. When $c = 0$, the Schr\"{o}dinger operator $L$ reduces to the classical Laplacian, whose localization has been discussed extensively \cite{nguyen2013localization}.

The aim of the present work is to investigate the high-frequency localization for the eigenfunctions of the Schr\"{o}dinger operator $L$ with a singular inverse square potential in high-dimensional balls and spherical shells. It turns out that the operator $L$ has an azimuthal quantum number $l$ and a principal quantum number $k$. Compared with previous studies on whispering gallery modes in the limit of $l\rightarrow\infty$ and focusing modes in the limit of $k\rightarrow\infty$ \cite{nguyen2013localization}, we consider here a novel two-parameter localization for Schr\"{o}dinger eigenfunctions as $l,k\rightarrow\infty$ simultaneously while keeping their ratio $l/k\rightarrow w$ as a constant. Some new localization phenomena are discovered and the rigorous $L^p$-localization theory is established. We prove that the eigenfunctions in balls are localized around an intermediate sphere whose radius is increasing with respect to the $l$-$k$ ratio $w$. The classical whispering gallery modes and focusing modes for the Laplaian can be viewed as limiting cases of our two-parameter localization as $w\rightarrow\infty$ and $w\rightarrow 0$, respectively. The eigenfunctions are shown to decay exponentially inside the localized sphere and decay polynomially outside. Furthermore, we observe an interesting second-order phase transition for the eigenfunctions in spherical shells when the $l$-$k$ ratio $w$ crosses a critical value $s(R)$ which is decreasing with respect to the ratio $R$ of the outer and inner radii of the spherical shell. In the supercritical case of $w>s(R)$, the eigenfunctions are localized around an intermediate sphere whose radius is increasing with respect to $w$ with $w = \infty$ corresponding to whispering gallery modes. In the critical case of $w = s(R)$, the eigenfunctions are localized around the inner boundary of the spherical shell, giving rise to the so-called critical modes. In the subcritical case of $w<s(R)$, the eigenfunctions fail to be localized, leading to new phenomena of localization breaking and focusing mode breaking.

\section{Localization for Schr\"{o}dinger eigenfunctions in balls}
In this section, we consider the case when the domain
\begin{equation*}
\Omega_{\textrm{ball}} = \{x\in\Rnum^d:|x|<1\}
\end{equation*}
is the unit ball centered at the origin. In spherical coordinates, the Laplace operator can be written as
\begin{equation*}
\Delta=\frac{\partial^2}{\partial r^2}+\frac{d-1}{r}\frac{\partial}{\partial r}+\frac{1}{r^2}\Delta_0
\end{equation*}
where $r=|x|$ is the radial coordinate and $\Delta_0$ is the Laplace-Beltrami operator on the $(d-1)$-dimensional sphere $S^{d-1}$. It is well known that the eigenfunctions of the spherical Laplacian $\Delta_0$ are the spherical harmonics of degree $l = 0,1,2,\cdots$ \cite{dai2013approximation, dunkl2014orthogonal}. For each $l\geq 0$, we have
\begin{equation*}
\Delta_0f = -l(l+d-2)f, \quad f\in H_l,
\end{equation*}
where $H_l$ is the vector space of spherical harmonics of degree $l$ whose dimension is given by
\begin{equation*}
\dim H_l = \binom{l+d-1}{d-1}-\binom{l+d-3}{d-1}.
\end{equation*}
In spherical coordinates, the eigenvalue problem \eqref{eigenproblem} can be rewritten as
\begin{equation*}
\frac{\partial^2u}{\partial r^2}+\frac{d-1}{r}\frac{\partial u}{\partial r}
+\left(\lambda+\frac{\Delta_0-c^2}{r^2}\right)u = 0.
\end{equation*}
We now represent the eigenfunction $u$ in the variable separation form as
\begin{equation*}
u(x) = v(r)Y_{lm}(\xi)
\end{equation*}
where $r$ is the angular coordinate, $\xi = (\xi_1,\cdots,\xi_{d-1})$ are angular coordinates, and $\{Y_{lm}: 1\leq m\leq\dim H_l\}$ is an orthonormal basis of $H_l$. It is easy to check that the radial part $v(r)$ satisfies the second-order ordinary differential equation
\begin{equation}\label{eq: InfODE}
\frac{\partial^2v}{\partial r^2}+\frac{d-1}{r}\frac{\partial v}{\partial r}
+\left(\lambda-\frac{l(l+d-2)+c^2}{r^2}\right)v = 0.
\end{equation}
For convenience, set
\begin{equation*}
\nu_l = \sqrt{\left(l+\frac{d}{2}-1\right)^2+c^2}.
\end{equation*}
Then \eqref{eq: InfODE} can be rewritten as
\begin{equation*}
r^2\frac{\partial^2}{\partial r^2}[r^{\frac{d}{2}-1}v]
+r\frac{\partial}{\partial r}[r^{\frac{d}{2}-1}v]+(\lambda r^2-\nu_l^2)r^{\frac{d}{2}-1}v = 0.
\end{equation*}
We then define a new variable $t = \sqrt{\lambda}r$ and set $\hat{v}(t) = r^{\frac{d}{2}-1}v(r)$. Then the new function $\hat{v}$ turns out to be the solution of the Bessel equation
\begin{equation*}
t^2\frac{\partial^2\hat{v}}{\partial t^2}
+t\frac{\partial\hat{v}}{\partial t}+(t^2-\nu_l^2)\hat{v} = 0,
\end{equation*}
whose general solutions can be expressed as the linear combination of Bessel functions of the first and second kinds. Since we now focus on the eigenvalue problem in the unit ball, the eigenfunction should not explode at the origin and thus
\begin{equation*}
\hat{v}(t) = J_{\nu_l}(t),
\end{equation*}
where $J_{\nu_l}$ is the Bessel function of the first kind of order $\nu_l$. Thus, the radial part
\begin{equation*}
v(r) = r^{1-\frac{d}{2}}J_{\nu_l}(\sqrt{\lambda}r),
\end{equation*}
is an ultraspherical Bessel function. With this expression, the Dirichlet boundary condition is converted into
\begin{equation*}
J_{\nu_l}(\sqrt{\lambda}) = 0.
\end{equation*}
For each $l\geq 0$, the eigenvalue problem \eqref{eigenproblem} has infinitely many positive eigenvalues
\begin{equation*}
\lambda_{lk} = j_{\nu_l,k}^2,\;\;\;k = 1,2,\cdots,
\end{equation*}
where $j_{\nu_l,k}$ is the $k$th zero of the Bessel function $J_{\nu_l}$. Finally, all basis eigenfunctions of the eigenvalue problem \eqref{eigenproblem} can be represented as
\begin{equation*}
u_{klm}(r,\xi) = r^{1-\frac{d}{2}}J_{\nu_l}(j_{\nu_l,k}r)Y_{lm}(\xi).
\end{equation*}
In quantum mechanics, $k$ is called the principal quantum number, $l$ is called the azimuthal quantum number, and $m$ is called the magnetic quantum number. In the two-dimensional case, the corresponding eigenfunctions in the unit disk are given by
\begin{equation*}
u_{kl1}(r,\theta) = J_{\nu_l}(j_{\nu_l,k}r)\cos l\theta,\;\;\;
u_{kl2}(r,\theta) = J_{\nu_l}(j_{\nu_l,k}r)\sin l\theta.
\end{equation*}

There is an apparent difference between the eigenfunctions for Laplacian and Schr\"{o}dinger operators in the two dimensional case. When $d = 2$, the order $\nu_l$ of the Bessel function reduces to
\begin{equation*}
\nu_l = \sqrt{l^2+c^2}.
\end{equation*}
In particular, we have $\nu_0 = c$. Since $J_0(0) = 1$ and $J_c(0) = 0$ for any $c>0$, when the azimuthal quantum number $l = 0$, the Schr\"{o}dinger eigenfunctions will collapse at the origin but the Laplacian eigenfunctions will not, as depicted in Fig. \ref{comparison}(a),(b). With the increase of $l$, the eigenfunctions for the two operators become increasingly similar to each other.

\begin{figure}
\centerline{\includegraphics[width=1\textwidth]{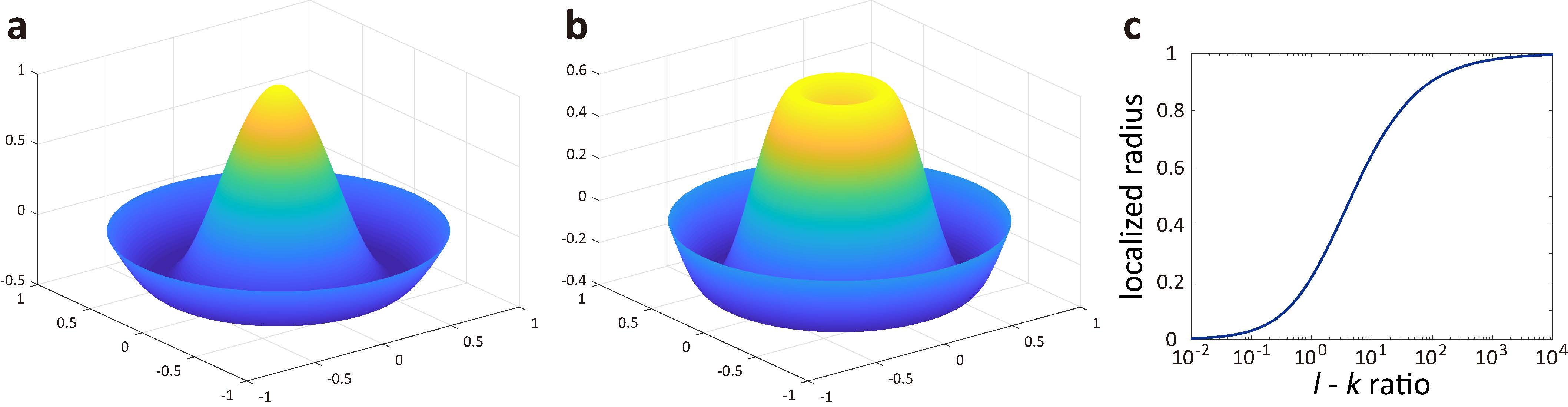}}
\caption{\textbf{Eigenfunctions and localized radii of the Laplacian and Schr\"{o}dinger operators.} (a) An eigenfunction of the two-dimensional Laplace operator when $l = 0$ and $k = 2$, which does not collapse at the origin. (b) An eigenfunction of the two-dimensional Schr\"{o}dinger operator when $c = 1$, $l = 0$, and $k = 2$, which collapses at the origin. (c) The localized radius of Schr\"{o}dinger eigenfunctions versus the $l$-$k$ ratio $w$ as $l,k\rightarrow\infty$ simultaneously.}\label{comparison}
\end{figure}

In previous studies, the high-frequency localization for Laplacian eigenfunctions has been studied extensively as $l\rightarrow\infty$ and $k$ is fixed or as $k\rightarrow\infty$ and $l$ is fixed \cite{nguyen2013localization}. Here we consider a more general two-parameter high-frequency localization for Schr\"{o}dinger eigenfunctions as $l,k\rightarrow\infty$ simultaneously while keeping $l/k\rightarrow w$.

\begin{lemma}
As $l,k\rightarrow\infty$ while keeping $l/k\rightarrow w>0$, we have
\begin{equation*}
\frac{j_{\nu_l,k}}{\nu_l} \rightarrow h(w),
\end{equation*}
where $h(w)>1$ is the unique solution of the algebraic equation
\begin{equation}\label{inverse}
\sqrt{h(w)^2-1}-\arcsec(h(w)) = \frac{\pi}{w}.
\end{equation}
\end{lemma}

\begin{proof}
Let $f$ be a function on $[1,\infty)$ defined by $f(z) = \sqrt{z^2-1}-\arcsec z$. It is easy to prove that $f$ is a strictly increasing function with
\begin{equation*}
f(1) = 0,\;\;\;\lim_{z\rightarrow\infty}f(z) = \infty.
\end{equation*}
Thus, the algebraic equation \eqref{inverse} has a unique solution $h(w)>1$. Let $a_k$ is the $k$th zero of the Airy function $\mathrm{Ai}(z)$. As $l,k\rightarrow\infty$ and $l/k\rightarrow w>0$, the zero $j_{\nu_l,k}$ of the Bessel function $J_{\nu_l}(z)$ has the following asymptotic behavior \cite[Equation 10.21.41]{special}:
\begin{equation*}
j_{\nu_l,k} = \nu_lz(\zeta)+O(\nu_l^{-1}),
\end{equation*}
where $\zeta = \nu_l^{-2/3}a_k$ and $z(\zeta)$ is the unique solution of
\begin{equation*}
\frac{2}{3}(-\zeta)^{3/2} = \sqrt{z(\zeta)^2-1}-\arcsec(z(\zeta)).
\end{equation*}
Moreover, as $k\rightarrow\infty$, the $k$th zero $a_k$ of the Airy function $\mathrm{Ai}(z)$ can be represented as \cite[Equation 9.9.6]{special}
\begin{equation*}
a_k = -T\left(\frac{3\pi}{8}(4k-1)\right),
\end{equation*}
where the function $T$ has the following asymptotic behavior \cite[Equation 9.9.18]{special}:
\begin{equation*}
T(t) = t^{2/3}+O(t^{-4/3}).
\end{equation*}
Combining the above two equations shows that
\begin{equation*}
a_k = -\left(\frac{3\pi k}{2}\right)^{2/3}+O(k^{-1/3}).
\end{equation*}
Therefore, we have
\begin{equation*}
-\zeta = -\nu_l^{-2/3}a_k \sim \left(\frac{3\pi k}{2l}\right)^{2/3}
\sim \left(\frac{3\pi}{2w}\right)^{2/3}.
\end{equation*}
This shows that
\begin{equation*}
\sqrt{z(\zeta)^2-1}-\arcsec(z(\zeta)) = \frac{2}{3}(-\zeta)^{3/2} \sim \frac{\pi}{w}.
\end{equation*}
Thus, we finally obtain that
\begin{equation*}
\frac{j_{\nu_l,k}}{\nu_l} = z(\zeta)+O(\nu_l^{-2}) \sim h(w),
\end{equation*}
which completes the proof.
\end{proof}

The two-parameter $L^\infty$-localization for Schr\"{o}dinger eigenfunctions in balls is stated as follows.

\begin{theorem}\label{linfty}
For any $r>0$ and $\epsilon>0$, let $D(r,\epsilon) = \{x\in\Rnum^d: ||x|-r|\geq\epsilon\}$. Then for any $\epsilon>0$, as $l,k\rightarrow\infty$ while keeping $l/k\rightarrow w>0$, we have
\begin{equation*}
\frac{\|u_{klm}\|_{L^\infty(D(1/h(w),\epsilon))}}{\|u_{klm}\|_{L^\infty(\Omega_{\textrm{ball}})}} \rightarrow 0,
\end{equation*}
where $h(w)>1$ is defined in \eqref{inverse}.
\end{theorem}

\begin{proof}
For any $r>0$ and $\epsilon>0$, let
\begin{equation}\label{shells}
D_1(r,\epsilon) = \{x\in\Rnum^d: |x|\leq r-\epsilon\},\;\;\;
D_2(r,\epsilon) = \{x\in\Rnum^d: r+\epsilon\leq|x|\leq 1\}.
\end{equation}
Clearly, we have
\begin{equation*}
\frac{\|u_{klm}\|_{L^\infty(D_1(1/h(w),\epsilon))}}{\|u_{klm}\|_{L^\infty(\Omega_{\textrm{ball}})}}
= \frac{\|(j_{\nu_l,k}r)^{1-d/2}J_{\nu_l}(j_{\nu_l,k}r)\|_{L^\infty([0,1/h(w)-\epsilon])}}
{\|(j_{\nu_l,k}r)^{1-d/2}J_{\nu_l}(j_{\nu_l,k}r)\|_{L^\infty([0,1])}}.
\end{equation*}
Since $j_{\nu,k}>\nu$ for any $\nu\geq 0$ and $k\geq 1$ \cite[Equation 10.21.3]{special}, we have
\begin{equation*}
\|(j_{\nu_l,k}r)^{1-d/2}J_{\nu_l}(j_{\nu_l,k}r)\|_{L^\infty([0,1])}
\geq \nu_l^{1-d/2}J_{\nu_l}(\nu_l).
\end{equation*}
Therefore, when $l$ is sufficiently large,
\begin{equation*}
\frac{\|u_{klm}\|_{L^\infty(D_1(1/h(w),\epsilon))}}{\|u_{klm}\|_{L^\infty(\Omega_{\textrm{ball}})}}
\leq \sup_{r\leq 1/h(w)-\epsilon}z(r)^{1-d/2}\frac{J_{\nu_l}(\nu_lz(r))}{J_{\nu_l}(\nu_l)},
\end{equation*}
where
\begin{equation*}
z(r) = \frac{j_{\nu_l,k}r}{\nu_l} < 1-\epsilon.
\end{equation*}
By the Paris inequality \cite[Equation 10.14.7]{special}, for any $\nu\geq 0$ and $0<z\leq 1$, we have
\begin{equation}\label{parisbessel}
1 \leq \frac{J_{\nu}(\nu z)}{z^\nu J_{\nu}(\nu)} \leq e^{\nu(1-z)}.
\end{equation}
Since $f(z) = ze^{1-z}$ is an increasing function on $[0,1]$, it is easy to see that
\begin{equation*}
\sup_{z<1-\epsilon}z^{1-d/2}\frac{J_{\nu_l}(\nu_lz)}{J_{\nu_l}(\nu_l)}
\leq e^{(1-z)(d/2-1)}[ze^{(1-z)}]^{\nu_l+1-d/2} \lesssim q^{\nu_l+1-d/2},
\end{equation*}
where
\begin{equation}\label{defq}
q = (1-\epsilon)e^\epsilon < 1.
\end{equation}
This shows that
\begin{equation}\label{part1}
\frac{\|u_{klm}\|_{L^\infty(D_1(1/h(w),\epsilon))}}{\|u_{klm}\|_{L^\infty(\Omega_{\textrm{ball}})}}
\lesssim q^{\nu_l+1-d/2} \rightarrow 0.
\end{equation}
On the other hand, when $l$ is sufficiently large,
\begin{equation*}
\frac{\|u_{klm}\|_{L^\infty(D_2(1/h(w),\epsilon))}}{\|u_{klm}\|_{L^\infty(\Omega_{\textrm{ball}})}}
\leq \sup_{1/h(w)+\epsilon\leq r\leq 1}z(r)^{1-d/2}
\frac{|J_{\nu_l}(\nu_lz(r))|}{J_{\nu_l}(\nu_l)},
\end{equation*}
where
\begin{equation}\label{temp1}
1+\epsilon < z(r) = \frac{j_{\nu_l,k}r}{\nu_l} < 2h(w).
\end{equation}
By a result of Cauchy \cite[Equation 10.19.8]{special}, as $\nu\rightarrow\infty$, we have
\begin{equation}\label{temp2}
J_{\nu}(\nu) \sim \frac{2^{1/3}}{3^{2/3}\Gamma(2/3)\nu^{1/3}}.
\end{equation}
Furthermore, as $\nu\rightarrow\infty$ with $0<\beta<\pi/2$ fixed, the Bessel function of the first kind has Debye's expansion \cite[Equation 10.19.6]{special}
\begin{equation*}
J_{\nu}(\nu\sec\beta) \sim \left(\frac{2}{\pi\nu\tan\beta}\right)^{1/2}
\left[\cos\left(\nu(\tan\beta-\beta)-\frac{\pi}{4}\right)+O(\nu^{-1})\right],
\end{equation*}
which holds uniformly for $\beta$ on compact sets in $(0,\pi/2)$. This suggests that as $\nu\rightarrow\infty$ with $z>1$ fixed, we have
\begin{equation}\label{temp3}
J_{\nu}(\nu z) \sim \left(\frac{4}{\pi^2(z^2-1)}\right)^{1/4}
\left[\cos\left(\zeta(z)\nu-\frac{\pi}{4}\right)+O(\nu^{-1})\right]\nu^{-1/2}.
\end{equation}
where
\begin{equation*}
\zeta(z) = \sqrt{z^2-1}-\mathrm{arcsec}(z).
\end{equation*}
Since the expansion holds uniformly for $z$ on compact sets in $(1,\infty)$, it follows from \eqref{temp1}, \eqref{temp2}, and \eqref{temp3} that
\begin{equation*}
\sup_{1+\epsilon<z<2h(w)}z^{1-d/2}\frac{|J_{\nu_l}(\nu_lz)|}{J_{\nu_l}(\nu_l)}
\lesssim \nu_l^{-1/6}.
\end{equation*}
This shows that
\begin{equation}\label{part2}
\frac{\|u_{klm}\|_{L^\infty(D_2(1/h(w),\epsilon))}}{\|u_{klm}\|_{L^\infty(\Omega_{\textrm{ball}})}}
\lesssim \nu_l^{-1/6} \rightarrow 0.
\end{equation}
Combining \eqref{part1} and \eqref{part2} gives the desired result.
\end{proof}

The two-parameter $L^p$-localization for Schr\"{o}dinger eigenfunctions in balls is stated as follows.

\begin{theorem}\label{lp}
The notation is the same as in Theorem \ref{linfty}. Then for any $\epsilon>0$ and $p>4$, as $l,k\rightarrow\infty$ while keeping $l/k\rightarrow w>0$, we have
\begin{equation*}
\frac{\|u_{klm}\|_{L^p(D(1/h(w),\epsilon))}}{\|u_{klm}\|_{L^p(\Omega_{\textrm{ball}})}} \rightarrow 0,
\end{equation*}
where $h(w)>1$ is defined in \eqref{inverse}.
\end{theorem}

\begin{proof}
For convenience, set $\alpha = (1-d/2)p+d-1$. For any $r>0$ and $\epsilon>0$, let $D_1(r,\epsilon)$ and $D_2(r,\epsilon)$ be the domains defined in \eqref{shells}. When $l$ is sufficiently large, we have
\begin{equation*}
\begin{split}
\frac{\|u_{klm}\|^p_{L^p(D_1(1/h(w),\epsilon))}}{\|u_{klm}\|^p_{L^p(\Omega_{\textrm{ball}})}}
&= \frac{\int_0^{1/h(w)-\epsilon}|r^{1-d/2}J_{\nu_l}(j_{\nu_l,k}r)|^pr^{d-1}dr}
{\int_0^1|r^{1-d/2}J_{\nu_l}(j_{\nu_l,k}r)|^pr^{d-1}dr} \\
&= \frac{\int_0^{(1/h(w)-\epsilon)j_{\nu_l,k}/\nu_l}z^\alpha|J_{\nu_l}(\nu_lz)|^pdz}
{\int_0^{j_{\nu_l,k}/\nu_l}z^\alpha|J_{\nu_l}(\nu_lz)|^pdz} \\
&\leq \frac{\int_0^{1-\epsilon}z^\alpha J_{\nu_l}(\nu_lz)^pdz}
{\int_0^1z^\alpha J_{\nu_l}(\nu_lz)^pdz}.
\end{split}
\end{equation*}
Using the Paris inequality \eqref{parisbessel} and noting that $f(z) = ze^{1-z}$ is an increasing function on $[0,1]$, we have
\begin{equation*}
\begin{split}
\int_0^{1-\epsilon}z^\alpha J_{\nu_l}(\nu_lz)^pdz
&\leq J_{\nu_l}(\nu_l)^p\int_0^{1-\epsilon}z^\alpha[ze^{(1-z)}]^{p\nu_l}dz \\
&\leq J_{\nu_l}(\nu_l)^p\int_0^{1-\epsilon}[ze^{(1-z)}]^{p\nu_l+\alpha}e^{-\alpha(1-z)}dz \\
&\lesssim J_{\nu_l}(\nu_l)^pq^{p\nu_l+\alpha}.
\end{split}
\end{equation*}
where $q\in(0,1)$ is the constant defined in \eqref{defq}. On the other hand, it follows from \cite[Equation 10.19.8]{special} that as $\nu\rightarrow\infty$ with $a\in\Rnum$ fixed, we have
\begin{equation*}
J_{\nu}(\nu+a\nu^{1/3}) = 2^{1/3}\Ai(-2^{1/3}a)\nu^{-1/3}[1+O(\nu^{-2/3})].
\end{equation*}
Since $\Ai(z)>0$ when $|z|$ is sufficiently small, there exists $a>0$ such that
\begin{equation}\label{middle}
J_{\nu}(z) \gtrsim \nu^{-1/3},\;\;\;\textrm{whenever\;}|z-\nu|\leq a\nu^{1/3}.
\end{equation}
This implies that
\begin{equation*}
\int_0^1z^\alpha J_{\nu_l}(\nu_lz)^pdz \geq \int_{1-a\nu_l^{-2/3}}^1z^\alpha J_{\nu_l}(\nu_lz)^pdz
\gtrsim \nu_l^{-(p+2)/3}.
\end{equation*}
Thus, for any $p\geq 1$, it follows from \eqref{temp2} that
\begin{equation}\label{lppart1}
\frac{\|u_{klm}\|^p_{L^p(D_1(1/h(w),\epsilon))}}{\|u_{klm}\|^p_{L^p(\Omega_{\textrm{ball}})}}
\lesssim \nu_l^{(p+2)/3}J_{\nu_l}(\nu_l)^pq^{p\nu_l+\alpha}
\lesssim \nu_l^{2/3}q^{p\nu_l+\alpha}\rightarrow 0.
\end{equation}
On the other hand, when $l$ is sufficiently large, we have
\begin{equation*}
\begin{split}
\frac{\|u_{klm}\|^p_{L^p(D_2(1/h(w),\epsilon))}}{\|u_{klm}\|^p_{L^p(\Omega_{\textrm{ball}})}}
&= \frac{\int_{1/h(w)+\epsilon}^1|r^{1-d/2}J_{\nu_l}(j_{\nu_l,k}r)|^pr^{d-1}dr}
{\int_0^1|r^{1-d/2}J_{\nu_l}(j_{\nu_l,k}r)|^pr^{d-1}dr} \\
&= \frac{\int_{(1/h(w)+\epsilon)j_{\nu_l,k}/\nu_l}^{j_{\nu_l,k}/\nu_l}z^\alpha|J_{\nu_l}(\nu_lz)|^pdz}
{\int_0^{j_{\nu_l,k}/\nu_l}z^\alpha|J_{\nu_l}(\nu_lz)|^pdz} \\
&\leq \frac{\int_{1+\epsilon}^{h(w)+\epsilon}z^\alpha|J_{\nu_l}(\nu_lz)|^pdz}
{\int_1^{h(w)-\epsilon}z^\alpha|J_{\nu_l}(\nu_lz)|^pdz}.
\end{split}
\end{equation*}
When $1+\epsilon\leq z\leq h(w)+\epsilon$, it follows from \eqref{temp3} that
\begin{equation*}
|J_{\nu_l}(\nu_lz)| \lesssim \nu_l^{-1/2}.
\end{equation*}
This shows that
\begin{equation*}
\begin{split}
\int_{1+\epsilon}^{h(w)+\epsilon}z^\alpha|J_{\nu_l}(\nu_lz)|^pdz
&\lesssim \nu_l^{-p/2}\int_{1+\epsilon}^{h(w)+\epsilon}z^\alpha dz \lesssim \nu_l^{-p/2}.
\end{split}
\end{equation*}
Moreover, it follows from \eqref{middle} that
\begin{equation*}
\int_1^{h(w)-\epsilon}z^\alpha|J_{\nu_l}(\nu_lz)|^pdz
\geq \int_1^{1+a\nu_l^{-2/3}}z^\alpha|J_{\nu_l}(\nu_lz)|^pdz \gtrsim \nu_l^{-(p+2)/3}.
\end{equation*}
Therefore, when $p>4$, we have
\begin{equation}\label{lppart2}
\frac{\|u_{klm}\|^p_{L^p(D_2(1/h(w),\epsilon))}}{\|u_{klm}\|^p_{L^p(\Omega_{\textrm{ball}})}}
\lesssim \nu_l^{(p+2)/3-p/2} \rightarrow 0.
\end{equation}
Combining \eqref{lppart1} and \eqref{lppart2} gives the desired result.
\end{proof}

\begin{remark}[\textbf{localized radii}]
The above two theorems characterize the two-parameter high-frequency localization for the Schr\"{o}dinger eigenfunctions $u_{klm}$ in the unit ball. As $l,k\rightarrow\infty$ while keeping their ratio $l/k\rightarrow w>0$ as a constant, the eigenfunctions are localized around an intermediate sphere with radius $1/h(w)\in (0,1)$. From \eqref{inverse}, we can see that the localized radius is independent of the strength $c$ of the inverse square potential. Fig. \ref{disk}(a)-(c) illustrate the graphs and heat maps of the eigenfunctions under different choices of $k$ and $l$ in the case of $d = 2$, $c = 1$, and $w = 5$, where the localized radius is $1/h(5)\approx 0.5$. The relationship between the $l$-$k$ ratio $w$ and the localized radius $1/h(w)$ is depicted in Fig. \ref{comparison}(c), from which we can see that with the increase of the $l$-$k$ ratio, the localized radius will also increase.
\end{remark}

\begin{figure}
\centerline{\includegraphics[width=1\textwidth]{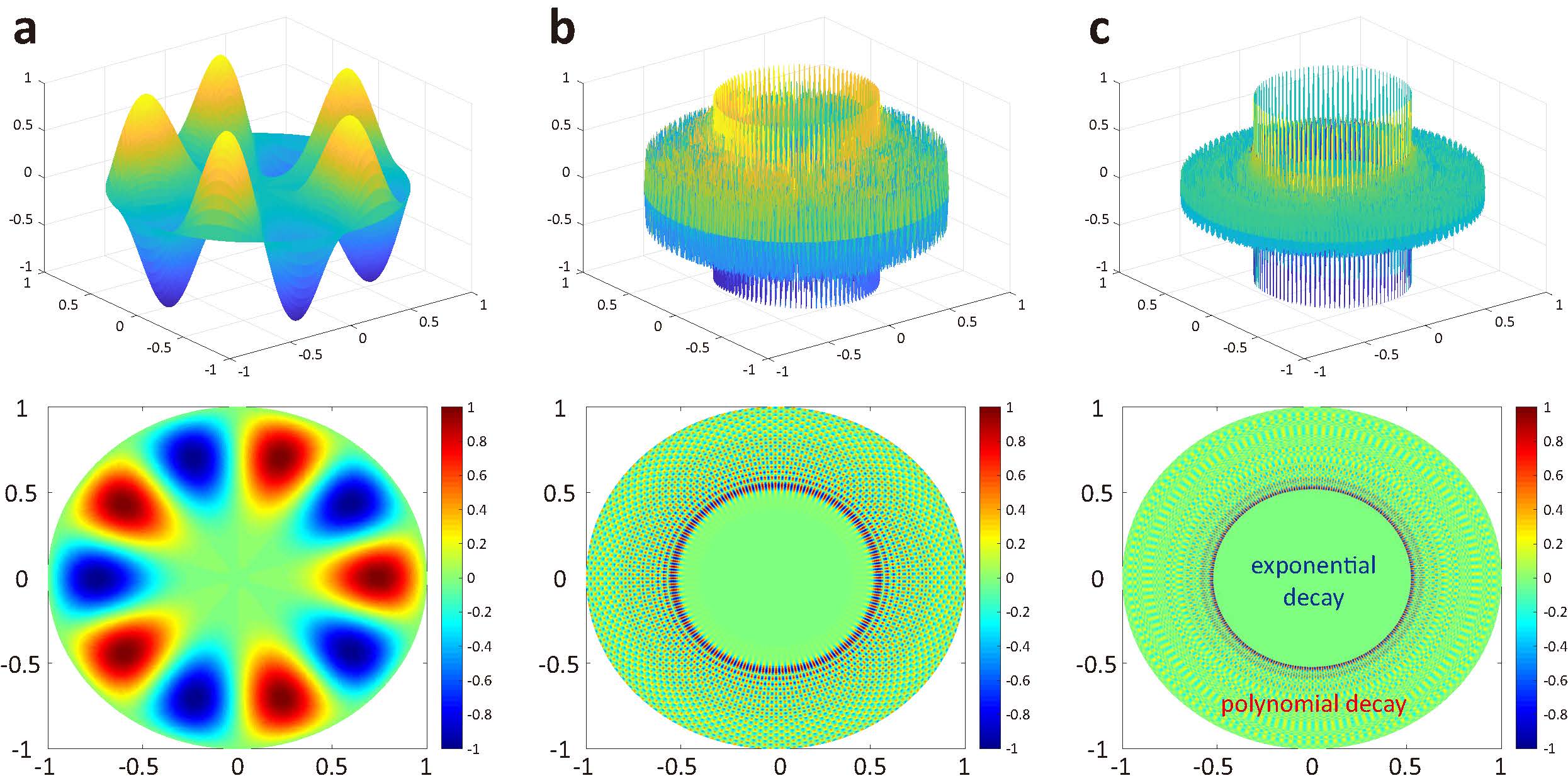}}
\caption{\textbf{Two-parameter localization for Schr\"{o}dinger eigenfunctions in the unit disk.} As $l,k\rightarrow\infty$ while keeping $l/k\rightarrow w>0$, the eigenfunctions $u_{klm}$ are localized around a sphere with radius $1/h(w)\in(0,1)$. The eigenfunctions are exhibited by both the three-dimensional graphs (upper) and two-dimensional heat maps with colorbars (lower). The eigenfunctions decay exponentially inside the localized circle and decay polynomially outside. (a) $l = 5$ and $k = 1$. (b) $l = 100$ and $k = 20$. (c) $l = 10000$ and $k = 200$. In (a)-(c), we keep the $l$-$k$ ratio as $w = 5$, in which case the localized radius is $1/h(5)\approx 0.5$. The parameter $c = 1$ in all cases and the eigenfunctions are normalized so that the supreme norm is $1$.}\label{disk}
\end{figure}

\begin{remark}[\textbf{decaying speed}]
Although the Schr\"{o}dinger eigenfunctions are mostly distributed around an intermediate sphere as $l,k\rightarrow\infty$ while keeping $l/k\rightarrow w$, the decaying speeds of the eigenfunctions inside and outside the localized sphere are completely different. From \eqref{part1} and \eqref{part2}, we can see that the eigenfunctions decay at the exponential speed of $q^{\nu_l}$ inside the localized sphere and decay at the much lower polynomial speed of $\nu_l^{-1/6}$ outside the localized sphere. This is clearly seen from Fig. \ref{disk}(c) in the two-dimensional case. When $l = 10000$ and $l/k = 5$, the eigenfunction almost vanishes inside the localized circle, but still fluctuates within a narrow range of zero outside the localized circle. In addition, from \eqref{lppart1} and \eqref{lppart2}, we can see that the eigenfunctions are $L^p$-localized for any $p\geq 1$ inside the localized sphere and are $L^p$-localized for any $p>4$ outside the localized sphere.
\end{remark}

Imitating the proof of the above two theorems or using the technique presented in \cite{nguyen2013localization}, we can obtain the one-parameter high-frequency localization for Schr\"{o}dinger eigenfunctions in balls as $l\rightarrow\infty$ and $k$ is fixed or as $k\rightarrow\infty$ and $l$ is fixed.

\begin{corollary}
For any $\epsilon>0$, let $A(\epsilon) = \{x\in\Rnum^d:|x|\leq 1-\epsilon\}$. Then for any $\epsilon>0$ and $1\leq p\leq\infty$, we have
\begin{equation*}
\lim_{l\rightarrow\infty}\frac{\|u_{klm}\|_{L^p(A(\epsilon))}}{\|u_{klm}\|_{L^p(\Omega_{\textrm{ball}})}} = 0.
\end{equation*}
\end{corollary}

The above corollary shows that the eigenfunctions are localized around the boundary of the unit ball as $l\rightarrow\infty$ and $k$ is fixed. These eigenfunctions are called whispering gallery modes \cite{nguyen2013localization}.

\begin{corollary}
For any $\epsilon>0$, let $B(\epsilon) = \{x\in\Rnum^d:\epsilon\leq |x|\leq 1\}$. Then for any $\epsilon>0$ and
\begin{equation*}
\frac{2d}{d-1}<p\leq\infty,
\end{equation*}
we have
\begin{equation*}
\lim_{k\rightarrow\infty}\frac{\|u_{klm}\|_{L^p(B(\epsilon))}}{\|u_{klm}\|_{L^p(\Omega_{\textrm{ball}})}} = 0.
\end{equation*}
\end{corollary}

\begin{proof}
We only consider $L^p$-localization for $p<\infty$. When $k$ is sufficiently large, we have
\begin{equation*}
\frac{\|u_{klm}\|^p_{L^p(B(\epsilon))}}{\|u_{klm}\|^p_{L^p(\Omega_{\textrm{ball}})}}
= \frac{\int_{\epsilon}^1|r^{1-d/2}J_{\nu_l}(j_{\nu_l,k}r)|^pr^{d-1}dr}
{\int_0^1|r^{1-d/2}J_{\nu_l}(j_{\nu_l,k}r)|^pr^{d-1}dr} \\
= \frac{\int_{\epsilon j_{\nu_l,k}}^{j_{\nu_l,k}}z^\alpha|J_{\nu_l}(z)|^pdz}
{\int_0^{j_{\nu_l,k}}z^\alpha|J_{\nu_l}(z)|^pdz},
\end{equation*}
where $\alpha = (1-d/2)p+d-1$. It is known that $j_{\nu,k}\sim m\pi$ as $m\rightarrow\infty$ for any $\nu\geq 0$ \cite[Equation 10.21.19]{special}. Moreover, as $z\rightarrow\infty$ with $\nu\geq 0$ fixed, it follows from \cite[Equation 10.17.3]{special} that
\begin{equation*}
J_{\nu}(z) \sim \left(\frac{2}{\pi z}\right)^{1/2}\cos\left(z-\frac{\nu\pi}{2}-\frac{\pi}{4}\right).
\end{equation*}
Thus, when $k$ is sufficiently large, we have
\begin{equation*}
\int_{\epsilon j_{\nu_l,k}}^{j_{\nu_l,k}}z^\alpha|J_{\nu_l}(z)|^pdz
\lesssim  \int_{\epsilon j_{\nu_l,k}}^{j_{\nu_l,k}}z^{\alpha-p/2}dz
\end{equation*}
When $p>2d/(d-1)$, we have $\alpha-p/2<-1$. It thus follows from Cauchy's criterion that
\begin{equation*}
\lim_{k\rightarrow\infty}\int_{\epsilon j_{\nu_l,k}}^{j_{\nu_l,k}}z^\alpha|J_{\nu_l}(z)|^pdz = 0.
\end{equation*}
Moreover, it is obvious that
\begin{equation*}
\lim_{k\rightarrow\infty}\int_0^{j_{\nu_l,k}}z^\alpha|J_{\nu_l}(z)|^pdz
= \int_0^\infty z^\alpha|J_{\nu_l}(z)|^pdz > 0.
\end{equation*}
Combining the above two equations gives the desired result.
\end{proof}

The above corollary shows that the eigenfunctions are localized around the center of the unit ball as $k\rightarrow\infty$ and $l$ is fixed. These eigenfunctions are called focusing modes \cite{nguyen2013localization}.

\begin{remark}[\textbf{relationship between one-parameter and two-parameter localization}]
In fact, both the whispering gallery modes and focusing modes can be viewed as limiting cases of our two-parameter localization. As $l\rightarrow\infty$ and $k$ is fixed, we have $l/k\rightarrow\infty$. In the limiting case of $w\rightarrow\infty$, the localized radius $1/h(w)\rightarrow 1$ and thus the eigenfunctions are localized around the boundary of the ball, giving rise to whispering gallery modes. On the other hand, as $k\rightarrow\infty$ and $l$ is fixed, we have $l/k\rightarrow 0$. In the limiting case of $w\rightarrow 0$, the localized radius $1/h(w)\rightarrow 0$ and thus the eigenfunctions are localized around the center of the ball, giving rise to focusing modes.
\end{remark}

\begin{remark}[\textbf{quantum mechanistic picture}]
From the aspect of quantum mechanics, as the azimuthal quantum number $l$ increases, the angular momentum of the particle also increases, which throws the particle outwards and thus induces localization around the boundary. On the other hand, since the inverse square potential $V(x) = c^2/|x|^2$ corresponds to an expulsive force, the energy of the particle is high when it is close to the origin and is low when it is far away from the origin. As the principal quantum number $k$ increases, the energy of the particle also increases, which pulls the particle towards the origin and thus induces localization around the center. These provide a qualitative explanation of whispering gallery modes and focusing modes.
\end{remark}

\section{Localization for Schr\"{o}dinger eigenfunctions in spherical shells}
In the section, we consider the eigenvalue problem \eqref{eigenproblem} in the domain
\begin{equation*}
\Omega_{\textrm{shell}} = \{x\in\Rnum^d:1<|x|<R\},
\end{equation*}
which is a spherical shell centered at the origin with inner radius 1 and outer radius $R>1$. Similarly, we represent the eigenfunction $u$ in the variable separation form as $u(x) = v(r)Y_{lm}(\xi)$, where $v(r)$ is the radial part and $Y_{lm}(\xi)$ is the angular part. We then define $t = \sqrt{\lambda}r$ and set $\hat{v}(t) = r^{\frac{d}{2}-1}v(r)$. Then the new function $\hat{v}$ turns out to be the solution of the Bessel equation
\begin{equation*}
t^2\frac{\partial^2\hat{v}}{\partial t^2}
+t\frac{\partial\hat{v}}{\partial t}+(t^2-\nu_l^2)\hat{v} = 0,
\end{equation*}
whose solution is given by
\begin{equation*}
\hat{v}(t) = \alpha J_{\nu_l}(t)+\beta Y_{\nu_l}(t),
\end{equation*}
where $J_{\nu_l}$ and $Y_{\nu_l}$ are Bessel functions of the first and second kinds of order $\nu_l$, respectively. Thus, the radial part of the eigenfunction is given by
\begin{equation*}
v(r) = r^{1-\frac{d}{2}}[\alpha J_{\nu_l}(\sqrt{\lambda}r)+\beta Y_{\nu_l}(\sqrt{\lambda}r)].
\end{equation*}
With this expression, the Dirichlet boundary condition is converted into the system of linear equations
\begin{equation*}\left\{
\begin{split}
& \alpha J_{\nu_l}(\sqrt{\lambda})+\beta Y_{\nu_l}(\sqrt{\lambda}) = 0,\\
& \alpha J_{\nu_l}(\sqrt{\lambda}R)+\beta Y_{\nu_l}(\sqrt{\lambda}R) = 0.
\end{split}\right.
\end{equation*}
Since $\alpha$ and $\beta$ are not simultaneously zero, the determinant of the coefficient matrix of the above system of linear equations must be zero, that is,
\begin{equation*}
J_{\nu_l}(\sqrt{\lambda})Y_{\nu_l}(\sqrt{\lambda}R)
-Y_{\nu_l}(\sqrt{\lambda})J_{\nu_l}(\sqrt{\lambda}R) = 0.
\end{equation*}
For each $l\geq 0$, the eigenvalue problem \eqref{eigenproblem} has infinitely many positive eigenvalues
\begin{equation*}
\lambda_{lk} = a_{\nu_l,k}^2,\;\;\;k = 1,2,\cdots,
\end{equation*}
where $a_{\nu_l,k}$ is the $k$th zero of the cross product
\begin{equation}\label{crossproduct}
f_{\nu_l,R}(z) = J_{\nu_l}(z)Y_{\nu_l}(Rz)-Y_{\nu_l}(z)J_{\nu_l}(Rz).
\end{equation}
Since $J_{\nu_l}$ and $Y_{\nu_l}$ do not have common positive zeros \cite[Equation 10.9.30]{special}, all basis eigenfunctions of the eigenvalue problem \eqref{eigenproblem} can be represented as
\begin{equation*}
u_{klm}(r,\xi)
= r^{1-\frac{d}{2}}[J_{\nu_l}(a_{\nu_l,k})Y_{\nu_l}(a_{\nu_l,k}r)
-Y_{\nu_l}(a_{\nu_l,k})J_{\nu_l}(a_{\nu_l,k}r)]Y_{lm}(\xi)
= r^{1-\frac{d}{2}}F_{\nu_l,k}(a_{\nu_l,k}r)Y_{lm}(\xi),
\end{equation*}
where
\begin{equation*}
F_{\nu_l,k}(z) = J_{\nu_l}(a_{\nu_l,k})Y_{\nu_l}(z)-Y_{\nu_l}(a_{\nu_l,k})J_{\nu_l}(z)
\end{equation*}
is a cylinder function, which is defined as a linear combination of Bessel functions of the first and second kinds.

We next investigate the two-parameter high-frequency localization for Schr\"{o}dinger eigenfunctions in spherical shells as $l,k\rightarrow\infty$ simultaneously while keeping $l/k\rightarrow w$.

\begin{lemma}\label{limittwosides}
As $l,k\rightarrow\infty$ while keeping $l/k\rightarrow w>0$, we have
\begin{equation*}
\frac{a_{\nu_l,k}}{\nu_l} \rightarrow \frac{g_R(w)}{w},
\end{equation*}
where $g_R(w)$ is the unique solution of the initial value problem
\begin{equation}\label{initial}
\frac{dy}{dw} =
\frac{\arccos\left(\frac{w}{Ry}\right)-\arccos\left(\frac{w}{y}\right)I_{\set{|w|\leq |y|}}}
{R\sqrt{1-\left(\frac{w}{Ry}\right)^2}-\sqrt{1-\left(\frac{w}{y}\right)^2}I_{\set{|w|\leq|y|}}},\;\;\;
y(0) = \frac{\pi}{R-1}.
\end{equation}
Moreover, both $g_R(w)$ and $w/g_R(w)$ are strictly increasing functions on $(0,\infty)$ with
\begin{equation*}
\lim_{w\rightarrow 0}\frac{w}{g_R(w)} = 0,\;\;\;
\lim_{w\rightarrow\infty}\frac{w}{g_R(w)} = R.
\end{equation*}
\end{lemma}

\begin{proof}
By \cite[Theorem 1.1]{bobkov2019asymptotic}, as $l,k\rightarrow\infty$ while keeping $l/k\rightarrow w>0$, we have
\begin{equation*}
\frac{a_{\nu_l,k}}{k} \rightarrow g_R(w),
\end{equation*}
where $g_R(w)$ is the unique solution of the initial value problem \eqref{initial}. This shows that
\begin{equation*}
\frac{a_{\nu_l,k}}{\nu_l} = \frac{a_{\nu_l,k}}{k}\frac{k}{\nu_l}\rightarrow \frac{g_R(w)}{w}.
\end{equation*}
Moreover, it follows from \cite[Proposition 3.1]{bobkov2019asymptotic} that $g_R(w)$ is a strictly increasing function with
\begin{equation}\label{twosides}
\frac{w}{R} < g_R(w) < \frac{f_R(w)}{R},
\end{equation}
where $f_R(w)$ is the unique solution to the initial value problem
\begin{equation}\label{initial2}
\frac{dy}{dw} = \frac{\arccos\left(\frac{w}{y}\right)}{\sqrt{1-\left(\frac{w}{y}\right)^2}},\;\;\;
y(0) = \frac{\pi R}{R-1}.
\end{equation}
Since
\begin{equation}\label{basic}
\arccos x = \arctan\frac{\sqrt{1-x^2}}{x} < \frac{\sqrt{1-x^2}}{x},\;\;\;0<x<1,
\end{equation}
it follows from \cite[Equation (3.1)]{bobkov2019asymptotic} that
\begin{equation*}
g_R'(w) \leq \frac{\arccos\left(\frac{w}{Rg_R(w)}\right)}{R\sqrt{1-\left(\frac{w}{Rg_R(w)}\right)^2}}
< \frac{g_R(w)}{w}.
\end{equation*}
Thus, we have
\begin{equation*}
\left(\frac{w}{g_R(w)}\right)' = \frac{g_R(w)-wg_R'(w)}{g_R(w)^2} > 0,
\end{equation*}
which shows that $w/g_R(w)$ is a strictly increasing function. Moreover, it is obvious that $f_R(w)>w$ for each $w>0$. It thus follows from \eqref{basic} that
\begin{equation}\label{auxilary}
f_R'(w) = \frac{\arccos\left(\frac{w}{f_R(w)}\right)}{\sqrt{1-\left(\frac{w}{f_R(w)}\right)^2}}
< \frac{f_R(w)}{w}.
\end{equation}
Similarly, we can prove that $w/f_R(w)$ is a strictly increasing function. By L'Hospital's rule, we have
\begin{equation*}
\lim_{w\rightarrow\infty}\frac{f_R(w)}{w} = \lim_{w\rightarrow\infty}f_R'(w) \triangleq \eta.
\end{equation*}
Taking $w\rightarrow\infty$ in \eqref{auxilary} yields
\begin{equation*}
\eta\sqrt{1-\left(1/\eta\right)^2} = \arccos\left(1/\eta\right),
\end{equation*}
whose unique solution is given by
\begin{equation}\label{limitfR}
\lim_{w\rightarrow\infty}\frac{f_R(w)}{w} = \eta = 1.
\end{equation}
This fact, together with \eqref{twosides}, gives the desired limits.
\end{proof}

The graphs of the function $w/g_R(w)$ under different choices of the outer radius $R>1$ are depicted in Fig. \ref{critical}(a). Since it is a strictly increasing function ranging from $0$ to $R$, there must exist a unique critical value $s(R)>0$ such that
\begin{equation*}
g_R(s(R)) = s(R).
\end{equation*}
This critical value will play a crucial role in the localization behavior for Schr\"{o}dinger eigenfunctions in spherical shells. The monotonic dependence of the critical value $s(R)$ on the outer radius $R$ is described in the following proposition and depicted in Fig. \ref{critical}(b).

\begin{figure}
\centerline{\includegraphics[width=1.0\textwidth]{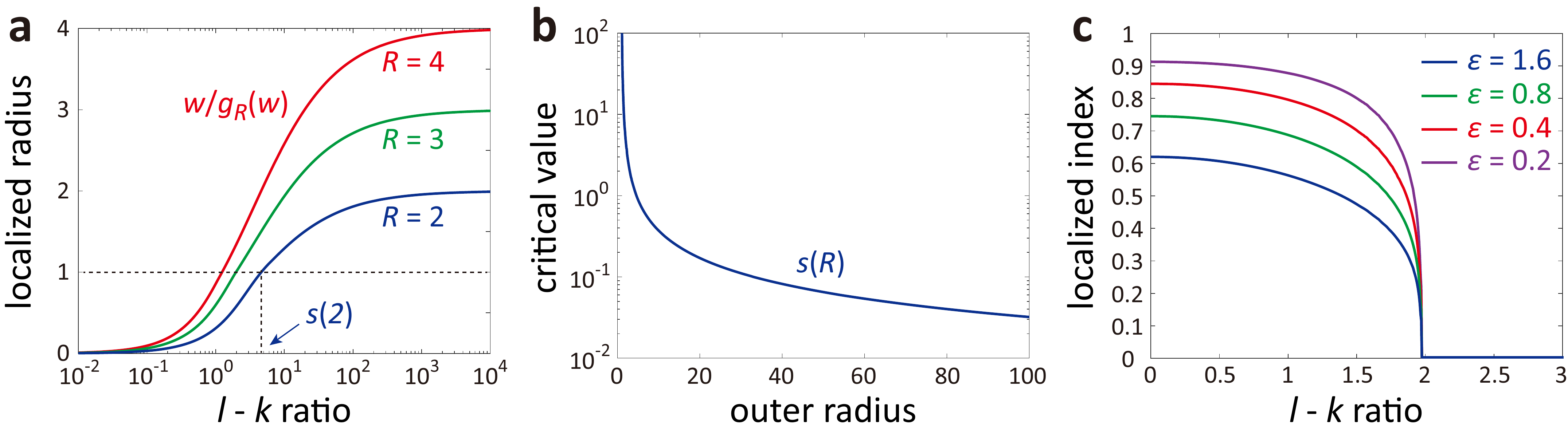}}
\caption{\textbf{Localized radii, critical values, and localized indices for Schr\"{o}dinger eigenfunctions in spherical shells.} (a) The localized radius $w/g_R(w)$ of the eigenfunctions versus the $l$-$k$ ratio $w$ under different choices of the outer radius $R$ as $l,k\rightarrow\infty$ simultaneously. (b) The critical value $s(R)$ versus the outer radius $R$. (c) The localized index $\gamma_{R,\epsilon}(w)$ versus the $l$-$k$ ratio $w$ in the case of $R = 3$ under different choices of $\epsilon$, which characterizes the width of the neighborhood of the localized sphere.}\label{critical}
\end{figure}

\begin{proposition}
$s(R)$ is a strictly decreasing function of $R$. Moreover, $s(R)$ has the following limit behavior:
\begin{equation*}
\lim_{R\rightarrow 1+}s(R) = \infty,\;\;\;\lim_{R\rightarrow\infty}s(R) = 0.
\end{equation*}
\end{proposition}

\begin{proof}
By the definition of $g_R(w)$, it is not hard to prove that $g_R(w)$ is strictly decreasing with respect to $R$ for each $w>0$. This shows that $s(R)$ is a strictly decreasing function of $R$. Moreover, it follows from \cite[Proposition 3.1]{bobkov2019asymptotic} that
\begin{equation}\label{gRinequality}
\frac{\pi}{R-1} < g_R(w) < \frac{\pi}{R-1}+\frac{\pi w}{2R}.
\end{equation}
This shows that
\begin{equation*}
s(R) > \frac{\pi}{R-1},
\end{equation*}
which gives the first limit. On the other hand, when $R>\pi/2$, it follows from \eqref{gRinequality} that
\begin{equation*}
s(R) < \frac{2\pi R}{(R-1)(2R-\pi)},
\end{equation*}
which gives the second limit.
\end{proof}

To investigate the localization behavior for Schr\"{o}dinger eigenfunctions in spherical shells, we need the Paris-type inequalities for cylinder functions. In fact, the lower bound of the Paris-type inequality for cylinder functions has been discussed in \cite{laforgia1986inequalities}. The following lemma gives an upper bound.

\begin{lemma}\label{paris}
Suppose that $l/k\rightarrow w>s(R)$ as $l,k\rightarrow\infty$. When $l$ is sufficiently large, we have
\begin{equation}\label{parisinequality}
0 \leq \frac{F_{\nu_l,k}(\nu_lz)}{z^{\nu_l}F_{\nu_l,k}(\nu_l)}
\leq e^{\frac{\nu_l^2(1-z^2)}{2(\nu_l+1)}} \leq e^{\nu_l(1-z)},\;\;\;\frac{a_{\nu_l,k}}{\nu_l}\leq z\leq 1.
\end{equation}
\end{lemma}

\begin{proof}
Since $z/g_R(z)$ is a strictly increasing function, as $l,k\rightarrow\infty$ while keeping $l/k\rightarrow w>s(R)$, we have
\begin{equation*}
\frac{a_{\nu_l,k}}{\nu_l} \rightarrow \frac{g_R(w)}{w} < \frac{g_R(s(R))}{s(R)} = 1.
\end{equation*}
When $l$ is sufficiently large, we have $a_{\nu_l,k}<\nu_l$ and thus
\begin{equation*}
Y_{\nu_l}(a_{\nu_l,k}) < 0 < J_{\nu_l}(a_{\nu_l,k}),\;\;\;Y_{\nu_l}(\nu_l) < 0 < J_{\nu_l}(\nu_l).
\end{equation*}
It follows from \cite[Page 80]{laforgia1986inequalities} that
\begin{equation*}
\frac{J_{\nu}(\nu z)}{z^\nu J_{\nu}(\nu)} \leq e^{\frac{\nu^2(1-z^2)}{2(\nu+1)}},\;\;\;\nu>0,\;0<z\leq 1.
\end{equation*}
and it follows from \cite[Page 79]{laforgia1986inequalities} that
\begin{equation*}
\frac{Y_{\nu}(\nu z)}{z^\nu Y_{\nu}(\nu)} \geq e^{\frac{\nu^2(1-z^2)}{2(\nu+1)}},\;\;\;\nu>0,\;0<z\leq 1.
\end{equation*}
Combining the above two equations, we obtain that
\begin{equation*}
J_{\nu_l}(a_{\nu_l,k})Y_{\nu_l}(\nu_lz)-Y_{\nu_l}(a_{\nu_l,k})J_{\nu_l}(\nu_lz)
\leq e^{\frac{\nu_l^2(1-z^2)}{2(\nu_l+1)}}z^{\nu_l}
[J_{\nu_l}(a_{\nu_l,k})Y_{\nu_l}(\nu_l)-Y_{\nu_l}(a_{\nu_l,k})J_{\nu_l}(\nu_l)].
\end{equation*}
For any $a_{\nu_l,k}/\nu_l < z\leq 1$, we have
\begin{equation*}
J_{\nu_l}(a_{\nu_l,k})Y_{\nu_l}(\nu_lz)-Y_{\nu_l}(a_{\nu_l,k})J_{\nu_l}(\nu_lz) = f_{\nu_l,R(z)}(a_{\nu_l,k}),
\end{equation*}
where $R(z) = \nu_lz/a_{\nu_l,k}>1$ and $f_{\nu_l,R(z)}$ is the cross product defined in \eqref{crossproduct}. Let $b_{\nu_l,1}$ be the first zero of the cross product $f_{\nu_l,R(z)}$. It then follows from \cite[Equation 1.1]{bobkov2019asymptotic} that
\begin{equation*}
b_{\nu_l,1} > \frac{\nu_l}{R(z)} = \frac{a_{\nu_l,k}}{z} \geq a_{\nu_l,k}.
\end{equation*}
This shows that $f_{\nu_l,R(z)}(a_{\nu_l,k}) > 0$ for any $a_{\nu_l,k}/\nu_l < z\leq 1$ and thus
\begin{equation*}
0 < \frac{F_{\nu_l,k}(\nu_lz)}{z^{\nu_l}F_{\nu_l,k}(\nu_l)}
= \frac{J_{\nu_l}(a_{\nu,k})Y_{\nu_l}(\nu_lz)-Y_{\nu_l}(a_{\nu_l,k})J_{\nu_l}(\nu_lz)}
{z^{\nu_l}[J_{\nu_l}(a_{\nu_l,k})Y_{\nu_l}(\nu_l)-Y_{\nu_l}(a_{\nu_l,k})J_{\nu_l}(\nu_l)]}
\leq e^{\frac{\nu_l^2(1-x^2)}{2(\nu_l+1)}},
\end{equation*}
which completes the proof.
\end{proof}

The two-parameter $L^\infty$-localization for Schr\"{o}dinger eigenfunctions in spherical shells is stated as follows.

\begin{theorem}\label{linfty2}
For any $r>0$ and $\epsilon>0$, let $D(r,\epsilon) = \{x\in\Rnum^d: ||x|-r|\geq\epsilon\}$. Then for any $\epsilon>0$, as $l,k\rightarrow\infty$ while keeping $l/k\rightarrow w>s(R)$, we have
\begin{equation*}
\frac{\|u_{klm}\|_{L^\infty(D(w/g_R(w),\epsilon))}}{\|u_{klm}\|_{L^\infty(\Omega_{\textrm{shell}})}} \rightarrow 0,
\end{equation*}
where $g_R(w)$ is defined in \eqref{initial}.
\end{theorem}

\begin{proof}
Since $z/g_R(z)$ is a strictly increasing function and $w>s(R)$, we have
\begin{equation*}
1 = \frac{s(R)}{g_R(s(R))} < \frac{w}{g_R(w)} < R.
\end{equation*}
For any $r>0$ and $\epsilon>0$, let $D_1(r,\epsilon)$ and $D_2(r,\epsilon)$ be the domains defined in \eqref{shells}. Clearly, we have
\begin{equation*}
\frac{\|u_{klm}\|_{L^\infty(D_1(w/g_R(w),\epsilon))}}{\|u_{klm}\|_{L^\infty(\Omega_{\textrm{shell}})}}
= \frac{\|(a_{\nu_l,k}r)^{1-d/2}F_{\nu_l,k}(a_{\nu_l,k}r)\|_{L^\infty([1,w/g_R(w)-\epsilon])}}
{\|(a_{\nu_l,k}r)^{1-d/2}F_{\nu_l,k}(a_{\nu_l,k}r)\|_{L^\infty([1,R])}}.
\end{equation*}
When $l$ is sufficiently large, we have
\begin{equation*}
1 < \frac{\nu_l}{a_{\nu_l,k}} < R,
\end{equation*}
and thus
\begin{equation*}
\|(a_{\nu_l,k}r)^{1-d/2}F_{\nu_l,k}(a_{\nu_l,k}r)\|_{L^\infty([1,R])}
\geq \nu_l^{1-d/2}F_{\nu_l,k}(\nu_l).
\end{equation*}
Therefore, when $l$ is sufficiently large,
\begin{equation*}
\frac{\|u_{klm}\|_{L^\infty(D_1(w/g_R(w),\epsilon))}}{\|u_{klm}\|_{L^\infty(\Omega_{\textrm{shell}})}}
\leq \sup_{1\leq r\leq w/g_R(w)-\epsilon}z(r)^{1-d/2}\frac{F_{\nu_l,k}(\nu_lz(r))}{F_{\nu_l,k}(\nu_l)},
\end{equation*}
where
\begin{equation*}
\frac{a_{\nu_l,k}}{\nu_l} \leq z(r) = \frac{a_{\nu_l,k}r}{\nu_l} < 1-\frac{\epsilon}{R}.
\end{equation*}
By Lemma \ref{paris}, since $f(x) = xe^{1-x}$ is an increasing function on $[0,1]$, it is easy to see that
\begin{equation*}
\sup_{z<1-\epsilon/R}z^{1-d/2}\frac{F_{\nu_l,k}(\nu_lz)}{F_{\nu_l,k}(\nu_l)}
\leq e^{(1-z)(d/2-1)}[ze^{(1-z)}]^{\nu_l+1-d/2} \lesssim q^{\nu_l+1-d/2},
\end{equation*}
where
\begin{equation}\label{defq2}
q = (1-\epsilon/R)e^{\epsilon/R} < 1.
\end{equation}
This shows that
\begin{equation}\label{part3}
\frac{\|u_{klm}\|_{L^\infty(D_1(w/g_R(w),\epsilon))}}{\|u_{klm}\|_{L^\infty(\Omega_{\textrm{shell}})}}
\lesssim q^{\nu_l+1-d/2} \rightarrow 0.
\end{equation}
On the other hand, when $l$ is sufficiently large,
\begin{equation*}
\frac{\|u_{klm}\|_{L^\infty(D_2(w/g_R(w),\epsilon))}}{\|u_{klm}\|_{L^\infty(\Omega_{\textrm{shell}})}}
\leq \sup_{w/g_R(w)+\epsilon\leq r\leq R}z(r)^{1-d/2}
\frac{|F_{\nu_l,k}(\nu_lz(r))|}{F_{\nu_l,k}(\nu_l)},
\end{equation*}
where
\begin{equation*}
1+\frac{\epsilon}{R} < z(r) = \frac{a_{\nu_l,k}r}{\nu_l} < R.
\end{equation*}
By Debye's expansions for Bessel functions \cite[Equation 10.19.3]{special}, as $\nu\rightarrow\infty$ with $\alpha>0$ fixed, we have
\begin{equation}\label{debye}
J_{\nu}(\nu\sech\alpha) \sim \frac{e^{-(\alpha-\tanh\alpha)\nu}}{(2\pi\nu\tanh\alpha)^{1/2}},\;\;\;
Y_{\nu}(\nu\sech\alpha) \sim -\frac{e^{(\alpha-\tanh\alpha)\nu}}
{\left(\frac{1}{2}\pi\nu\tanh\alpha\right)^{1/2}}.
\end{equation}
Moreover, as $\nu\rightarrow\infty$, it follows from \cite[Equation 10.19.8]{special} that
\begin{equation*}
J_{\nu}(\nu) \sim \frac{2^{1/3}}{3^{2/3}\Gamma(2/3)\nu^{1/3}},\;\;\;
Y_{\nu}(\nu) \sim -\frac{2^{1/3}}{3^{1/6}\Gamma(2/3)\nu^{1/3}}.
\end{equation*}
This shows that
\begin{equation}\label{temp5}
F_{\nu_l,k}(\nu_l) = J_{\nu_l}(a_{\nu_l,k})Y_{\nu_l}(\nu_l)-Y_{\nu_l}(a_{\nu_l,k})J_{\nu_l}(\nu_l)
\sim Ce^{(\alpha-\tanh\alpha)\nu}\nu_l^{-5/6},
\end{equation}
where
\begin{equation*}
C = \frac{2^{1/3}}{3^{1/6}\Gamma(2/3)\left(\frac{1}{2}\pi\tanh\alpha\right)^{1/2}} > 0
\end{equation*}
and $\alpha>0$ is the unique solution to the algebraic equation
\begin{equation*}
\sech\alpha = \frac{g_R(w)}{w} \in \left(\frac{1}{R},1\right).
\end{equation*}
Imitating the proof of \eqref{temp3} and applying Debye's expansion \cite[Equation 10.19.6]{special}, we obtain that
\begin{equation}\label{temp4}
Y_{\nu}(\nu z) \sim \left(\frac{4}{\pi^2(z^2-1)}\right)^{1/4}
\left[\sin\left(\zeta(z)\nu-\frac{\pi}{4}\right)+O(\nu^{-1})\right]\nu^{-1/2},
\end{equation}
which holds uniformly for $z$ on compact intervals in $(1,\infty)$. It thus follows from \eqref{temp3} and \eqref{temp4} that
\begin{equation}\label{temp6}
F_{\nu_l,k}(\nu_lz) = J_{\nu_l}(a_{\nu_l,k})Y_{\nu_l}(\nu_lz)-Y_{\nu_l}(a_{\nu_l,k})J_{\nu_l}(\nu_lz)
= e^{(\alpha-\tanh\alpha)\nu}O(\nu_l^{-1}),
\end{equation}
which holds uniformly for $z$ on compact intervals in $(1,\infty)$. Thus, we finally obtain that
\begin{equation}\label{part4}
\frac{\|u_{klm}\|_{L^\infty(D_2(w/g_R(w),\epsilon))}}{\|u_{klm}\|_{L^\infty(\Omega_{\textrm{shell}})}}
\leq \sup_{1+\epsilon/R<z<R}z^{1-d/2}\frac{|F_{\nu_l,k}(\nu_lz)|}{F_{\nu_l,k}(\nu_l)}
\lesssim \nu_l^{-1/6} \rightarrow 0.
\end{equation}
Combining \eqref{part3} and \eqref{part4} gives the desired result.
\end{proof}

The two-parameter $L^p$-localization for Schr\"{o}dinger eigenfunctions in spherical shells is stated as follows.

\begin{theorem}\label{lp2}
The notation is the same as in Theorem \ref{linfty2}. Then for any $\epsilon>0$ and $p>4$, as $l,k\rightarrow\infty$ while keeping $l/k\rightarrow w>s(R)$, we have
\begin{equation*}
\frac{\|u_{klm}\|_{D(w/g_R(w),\epsilon)}}{\|u_{klm}\|_{\Omega_{\textrm{shell}}}} \rightarrow 0,
\end{equation*}
where $g_R(w)$ is defined in \eqref{initial}.
\end{theorem}

\begin{proof}
For convenience, set $\alpha = (1-d/2)p+d-1$. For any $r>0$ and $\epsilon>0$, let $D_1(r,\epsilon)$ and $D_2(r,\epsilon)$ be the domains defined in \eqref{shells}. When $l$ is sufficiently large, we have
\begin{equation*}
\begin{split}
\frac{\|u_{klm}\|^p_{L^p(D_1(w/g_R(w),\epsilon))}}{\|u_{klm}\|^p_{L^p(\Omega_{\textrm{shell}})}}
&= \frac{\int_1^{w/g_R(w)-\epsilon}|r^{1-d/2}F_{\nu_l,k}(a_{\nu_l,k}r)|^pr^{d-1}dr}
{\int_1^R|r^{1-d/2}F_{\nu_l,k}(a_{\nu_l,k}r)|^pr^{d-1}dr} \\
&= \frac{\int_{a_{\nu_l,k}/\nu_l}^{(w/g_R(w)-\epsilon)a_{\nu_l,k}/\nu_l}z^\alpha|F_{\nu_l,k}(\nu_lz)|^pdz}
{\int_{a_{\nu_l,k}/\nu_l}^{Ra_{\nu_l,k}/\nu_l}z^\alpha|F_{\nu_l,k}(\nu_lz)|^pdz} \\
&\leq \frac{\int_{a_{\nu_l,k}/\nu_l}^{1-\epsilon/R}z^\alpha|F_{\nu_l,k}(\nu_lz)|^pdz}
{\int_{a_{\nu_l,k}/\nu_l}^1z^\alpha F_{\nu_l,k}(\nu_lz)^pdz}.
\end{split}
\end{equation*}
Using the Paris-type inequality \eqref{parisinequality} and noting that $f(z) = ze^{1-z}$ is an increasing function on $[0,1]$, we have
\begin{equation*}
\begin{split}
\int_{a_{\nu_l,k}/\nu_l}^{1-\epsilon/R}z^\alpha J_{\nu_l}(\nu_lz)^pdz
&\leq F_{\nu_l,k}(\nu_l)^p\int_{a_{\nu_l,k}/\nu_l}^{1-\epsilon/R}z^\alpha[ze^{(1-z)}]^{p\nu_l}dz \\
&\leq F_{\nu_l,k}(\nu_l)^p\int_{a_{\nu_l,k}/\nu_l}^{1-\epsilon/R}
[ze^{(1-z)}]^{p\nu_l+\alpha}e^{-\alpha(1-z)}dz \\
&\lesssim F_{\nu_l,k}(\nu_l)^pq^{p\nu_l+\alpha}.
\end{split}
\end{equation*}
where $q\in(0,1)$ is the constant defined in \eqref{defq2}. Moreover, it follows from \cite[Equation 10.19.8]{special} that as $\nu\rightarrow\infty$ with $a\in\Rnum$ fixed, we have
\begin{gather*}
J_{\nu}(\nu+a\nu^{1/3}) = 2^{1/3}\Ai(-2^{1/3}a)\nu^{-1/3}[1+O(\nu^{-2/3})], \\
Y_{\nu}(\nu+a\nu^{1/3}) = -2^{1/3}\Bi(-2^{1/3}a)\nu^{-1/3}[1+O(\nu^{-2/3})],
\end{gather*}
where $\Ai(z)$ and $\Bi(z)$ are Airy functions. Since $\Ai(z)>0$ and $\Bi(z)>0$ when $|z|$ is sufficiently small, there exists $a>0$ such that
\begin{equation*}
J_{\nu}(z) \gtrsim \nu^{-1/3},\;\;\;Y_{\nu}(z) \gtrsim -\nu^{-1/3},\;\;\;
\textrm{for any\;}|z-\nu|\leq a\nu^{1/3}.
\end{equation*}
Thus, whenever $|z-\nu|\leq a\nu^{1/3}$, it follows from Debye's expansions \eqref{debye} that
\begin{equation}\label{middle2}
F_{\nu_l,k}(z) = J_{\nu_l}(a_{\nu_l,k})Y_{\nu_l}(z)-Y_{\nu_l}(a_{\nu_l,k})J_{\nu_l}(z)
\gtrsim e^{(\alpha-\tanh\alpha)\nu_l}\nu_l^{-5/6},
\end{equation}
where $\alpha$ is the unique solution to the algebraic equation $\sech\alpha = g_R(w)/w$. This implies that
\begin{equation*}
\int_{a_{\nu_l,k}/\nu_l}^1z^\alpha F_{\nu_l,k}(\nu_lz)^pdz
\geq \int_{1-a\nu_l^{-2/3}}^{1}z^\alpha F_{\nu_l,k}(\nu_lz)^pdz
\gtrsim e^{(\alpha-\tanh\alpha)p\nu_l}\nu_l^{-(5p/6+2/3)}.
\end{equation*}
Thus, for any $p\geq 1$, it follows from \eqref{temp5} that
\begin{equation}\label{lppart3}
\frac{\|u_{klm}\|^p_{L^p(D_1(w/g_R(w),\epsilon))}}{\|u_{klm}\|^p_{L^p(\Omega_{\textrm{shell}})}}
\lesssim e^{-(\alpha-\tanh\alpha)p\nu_l}\nu_l^{5p/6+2/3}F_{\nu_l,k}(\nu_l)^pq^{p\nu_l+\alpha}
\lesssim \nu_l^{2/3}q^{p\nu_l+\alpha}\rightarrow 0.
\end{equation}
On the other hand, when $l$ is sufficiently large, we have
\begin{equation*}
\begin{split}
\frac{\|u_{klm}\|^p_{L^p(D_2(w/g_R(w),\epsilon))}}{\|u_{klm}\|^p_{L^p(\Omega_{\textrm{shell}})}}
&= \frac{\int_{w/g_R(w)+\epsilon}^R|r^{1-d/2}F_{\nu_l,k}(a_{\nu_l,k}r)|^pr^{d-1}dr}
{\int_1^R|r^{1-d/2}F_{\nu_l,k}(a_{\nu_l,k}r)|^pr^{d-1}dr} \\
&= \frac{\int_{(w/g_R(w)+\epsilon)a_{\nu_l,k}/\nu_l}^{Ra_{\nu_l,k}/\nu_l}z^\alpha|F_{\nu_l,k}(\nu_lz)|^pdz}
{\int_{a_{\nu_l,k}/\nu_l}^{Ra_{\nu_l,k}/\nu_l}z^\alpha|F_{\nu_l,k}(\nu_lz)|^pdz} \\
&\leq \frac{\int_{1+\epsilon/R}^Rz^\alpha|F_{\nu_l,k}(\nu_lz)|^pdz}
{\int_1^{1+\epsilon}z^\alpha|F_{\nu_l,k}(\nu_lz)|^pdz}
\end{split}
\end{equation*}
When $1+\epsilon/R\leq z\leq R$, it follows from \eqref{temp6} that
\begin{equation*}
|F_{\nu_l,k}(\nu_lz)| \lesssim e^{(\alpha-\tanh\alpha)\nu_l}\nu_l^{-1}.
\end{equation*}
This shows that
\begin{equation*}
\begin{split}
\int_{1+\epsilon/R}^Rz^\alpha|F_{\nu_l,k}(\nu_lz)|^pdz
&\lesssim e^{(\alpha-\tanh\alpha)p\nu_l}\nu_l^{-p}.
\end{split}
\end{equation*}
Moreover, it follows from \eqref{middle2} that
\begin{equation*}
\int_1^{1+\epsilon}z^\alpha|F_{\nu_l,k}(\nu_lz)|^pdz
\geq \int_1^{1+a\nu_l^{-2/3}}z^\alpha|F_{\nu_l,k}(\nu_lz)|^pdz
\gtrsim e^{(\alpha-\tanh\alpha)p\nu_l}\nu_l^{-(5p/6+2/3)}.
\end{equation*}
Therefore, when $p>4$, we have
\begin{equation}\label{lppart4}
\frac{\|u_{klm}\|^p_{L^p(D_2(1/h(w),\epsilon))}}{\|u_{klm}\|^p_{L^p(\Omega_{\textrm{ball}})}}
\lesssim \nu_l^{(5p/6+2/3)-p} \rightarrow 0.
\end{equation}
Combining \eqref{lppart3} and \eqref{lppart4} gives the desired result.
\end{proof}

\begin{remark}[\textbf{dynamical phase transition}]
The above two theorems characterize the two-parameter high-frequency localization for the Schr\"{o}dinger eigenfunctions $u_{klm}$ in a spherical shell as $l,k\rightarrow\infty$ simultaneously while keeping $l/k\rightarrow w$. It turns out that the $l$-$k$ ratio $w$ has a critical value $s(R)$, which separates the parameter region into two phases. Interestingly, we observe a dynamical phase transition when $w$ crosses its critical value.

In the supercritical case of $w>s(R)$, the eigenfunctions are localized around an intermediate sphere with radius $w/g_R(w)\in (1,R)$. From \eqref{initial}, we can see that the localized radius is independent of the strength $c$ of the inverse square potential. Fig. \ref{annulus}(a)-(c) illustrate the graphs and heat maps of the eigenfunctions in the supercritical case of $d = 2$, $c = 1$, $R = 3$, and $w = 10 > s(R) \approx 1.972$ under different choices of $k$ and $l$, where the localized radius is $w/g_R(w)\approx 2$. In analogy to eigenfunctions in balls, the eigenfunctions in a spherical shell decay exponentially inside the localized sphere and decay polynomially outside the localized sphere. In addition, the eigenfunctions are $L^p$-localized for any $p\geq 1$ inside the localized sphere and are $L^p$-localized for any $p>4$ outside the localized sphere.

In the subcritical case of $w<s(R)$, according to our numerical simulations, the eigenfunctions peak around the inner boundary of the spherical shell when $l$ and $k$ are large. However, localization fails to be observed in this case. This can be seen from Fig. \ref{annulus}(d)-(f), which depict the graphs and heat maps of the eigenfunctions in the subcritical case of $d = 2$, $c = 1$, $R = 3$, and $w = 1 < s(R) \approx 1.972$ under different choices of $k$ and $l$. This localization breaking can be explained as follows. If $w<s(R)$, we have $w/g_R(w)<1$  and thus the sphere with radius $w/g_R(w)$ is outside the spherical shell, which hinders the formation of localization.

To further understand the phase transition, we introduce a quantity called the localization index, which is defined as
\begin{equation*}
\gamma_{R,\epsilon}(w) = \left\{\begin{split} \lim_{l,k\rightarrow\infty}\frac{\|u_{klm}\|_{L^\infty(D(w/g_R(w),\epsilon))}}
{\|u_{klm}\|_{L^\infty(\Omega_{\textrm{shell}})}},\;\;\;x>s(R),\\
\lim_{l,k\rightarrow\infty}\frac{\|u_{klm}\|_{L^\infty(D(1,\epsilon))}}
{\|u_{klm}\|_{L^\infty(\Omega_{\textrm{shell}})}},
\;\;\;\;\;\;\;x\leq s(R). \end{split}\right.
\end{equation*}
which is a number between $0$ and $1$. It characterizes the maximum relative height of an eigenfunction outside an $\epsilon$-neighbourhood of the localized sphere. In the subcritical case, the localized sphere is chosen to be the inner boundary of the spherical shell because the eigenfunction peaks around the inner boundary when $l$ and $k$ are large. Obviously, $L^\infty$-localization occurs if and only if the localization index vanishes for any $\epsilon>0$. Fig. \ref{critical}(c) depicts the localization index $\gamma_{R,\epsilon}(w)$ versus the $l$-$k$ ratio $w$ under different choices of $\epsilon$ in the case of $R = 3$, where the critical value is $s(R)\approx 1.972$. Clearly, the system undergoes a second-order phase transition as $w$ cross its critical value $s(R)$. Here $w$ serves as the tuning parameter, $\gamma_{R,\epsilon}(w)$ serves as the order parameter, and $l$ or $k$ serves as the system size. It the subcritical case of $w < s(R)$, the localization index is always positive and thus localization of the eigenfunctions fails to be observed. In the critical case of $w = s(R)$, the localization index is zero, according to our numerical simulations. This suggests that the eigenfunctions are localized around the inner boundary of the spherical shell in the critical case.
\end{remark}

\begin{figure}
\centerline{\includegraphics[width=0.95\textwidth]{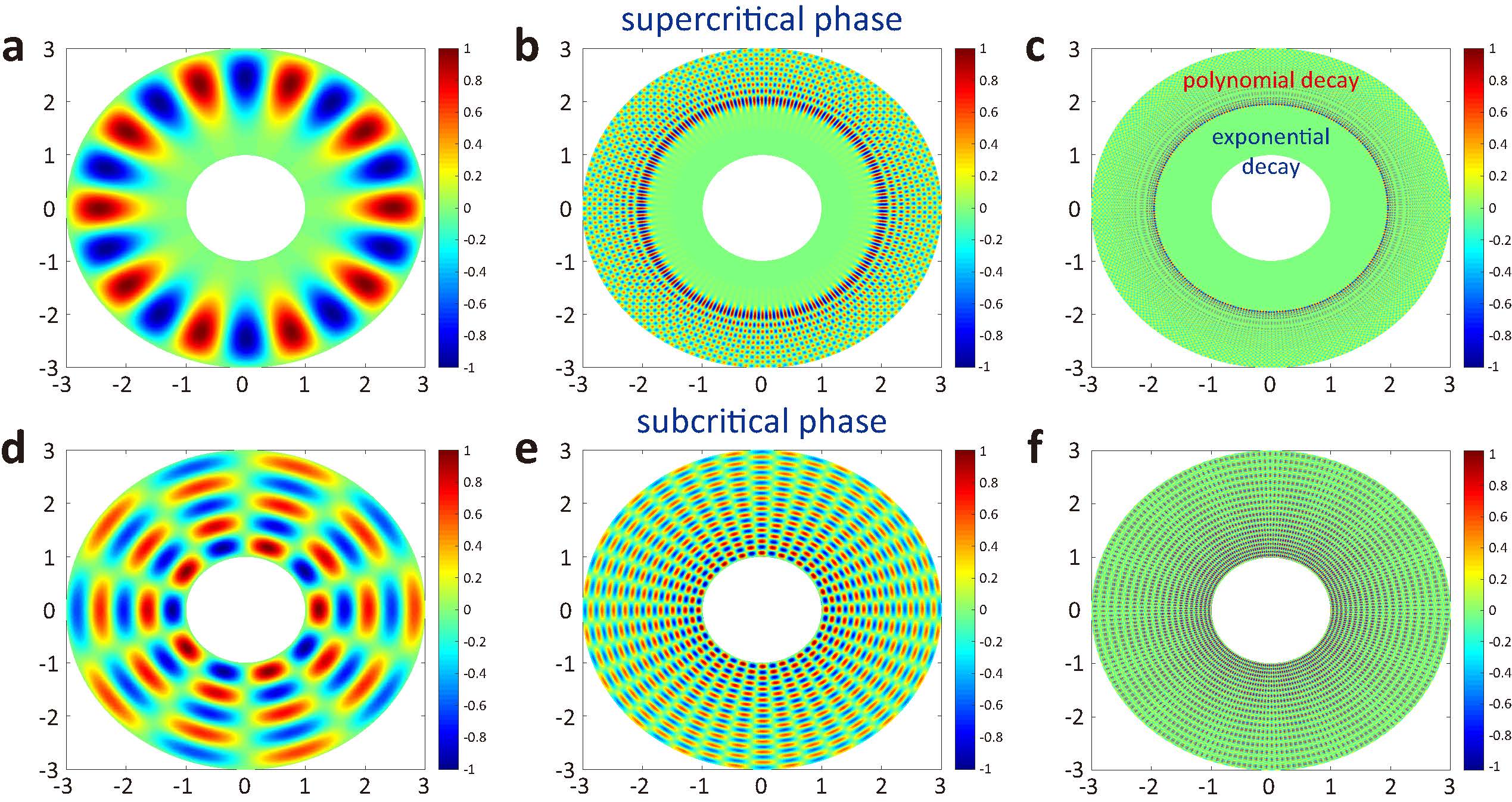}}
\caption{\textbf{Dynamical phase transition for Schr\"{o}dinger eigenfunctions in an annulus.} (a)-(c) Heat maps of the eigenfunctions under different choices of $l$ and $k$ in the supercritical case, where the outer radius is chosen as $R = 3$ and the $l$-$k$ ratio is chosen as $w = 10 > s(R)\approx 1.972$. When $l$ and $k$ are large, the eigenfunctions are localized around the circle with radius $w/g_R(w)\approx 2$. The eigenfunctions decay exponentially inside the localized circle and decay polynomially outside. (a) $k = 1$. (b) $k = 10$. (c) $k = 100$. (d)-(f) Heat maps of the eigenfunctions under different choices of $l$ and $k$ in the subcritical case, where the outer radius is chosen as $R = 3$ and the $l$-$k$ ratio is chosen as $w = 1$. The eigenfunctions fail to be localized in this case. (d) $k = 10$. (e) $k = 20$. (f) $k = 50$. The parameter $c = 1$ in all cases and the eigenfunctions are normalized so that the supreme norm is $1$.}\label{annulus}
\end{figure}

We next consider the one-parameter high-frequency localization for Schr\"{o}dinger eigenfunctions in spherical shells as $l\rightarrow\infty$ and $k$ is fixed. To this end, we need the following lemma.

\begin{lemma}\label{limitoneside}
For each $l\geq 0$ and $k\geq 1$, we have
\begin{equation*}
\frac{a_{\nu_l,k}}{\nu_l} > \frac{1}{R}.
\end{equation*}
Moreover, we have
\begin{equation*}
\lim_{l\rightarrow\infty}\frac{a_{\nu_l,k}}{\nu_l} = \frac{1}{R}.
\end{equation*}
\end{lemma}

\begin{proof}
It follows from \cite[Equations (1.1) and (4.5)]{bobkov2019asymptotic} that
\begin{equation*}
\frac{\nu_l}{R} < a_{\nu_l,k} < \frac{k}{R}f_R\left(\frac{\nu_l}{k}\right),
\end{equation*}
where $f_R(w)$ is the unique solution of the initial value problem \eqref{initial2}. This inequality, together with \eqref{limitfR}, gives the desired result.
\end{proof}

Imitating the proof of Theorem \ref{linfty2} and applying the above lemma, we obtain the one-parameter high-frequency localization for Schr\"{o}dinger eigenfunctions in spherical shells as $l\rightarrow\infty$.

\begin{corollary}\label{criticalmode}
For each $\epsilon>0$, let $A_R(\epsilon) = \{x\in\Rnum^d:1\leq |x|\leq R-\epsilon\}$. Then for any $\epsilon>0$ and $1\leq p\leq\infty$, we have
\begin{equation*}
\lim_{l\rightarrow\infty}\frac{\|u_{klm}\|_{L^p(A_R(\epsilon))}}
{\|u_{klm}\|_{L^p(\Omega_{\textrm{shell}})}} = 0.
\end{equation*}
\end{corollary}

\begin{proof}
We only consider $L^\infty$-localization. Clearly, we have
\begin{equation*}
\frac{\|u_{klm}\|_{L^\infty(A_R(\epsilon))}}{\|u_{klm}\|_{L^\infty(\Omega_{\textrm{shell}})}}
= \frac{\|(a_{\nu_l,k}r)^{1-d/2}F_{\nu_l,k}(a_{\nu_l,k}r)\|_{L^\infty([1,R-\epsilon])}}
{\|(a_{\nu_l,k}r)^{1-d/2}F_{\nu_l,k}(a_{\nu_l,k}r)\|_{L^\infty([1,R])}}.
\end{equation*}
When $l$ is sufficiently large, we have $\nu_l/a_{\nu_l,k} < R$ and thus
\begin{equation*}
\|(a_{\nu_l,k}r)^{1-d/2}F_{\nu_l,k}(a_{\nu_l,k}r)\|_{L^\infty([1,R])}
\geq \nu_l^{1-d/2}F_{\nu_l,k}(\nu_l).
\end{equation*}
Therefore, when $l$ is sufficiently large, it follows from Lemma \ref{limitoneside} that
\begin{equation*}
\frac{\|u_{klm}\|_{L^\infty(A_R(\epsilon))}}{\|u_{klm}\|_{L^\infty(\Omega_{\textrm{shell}})}}
\leq \sup_{1\leq r\leq R-\epsilon}z(r)^{1-d/2}\frac{F_{\nu_l,k}(\nu_lz(r))}{F_{\nu_l,k}(\nu_l)},
\end{equation*}
where
\begin{equation*}
z(r) = \frac{a_{\nu_l,k}r}{\nu_l} < 1-\frac{1}{2R}\epsilon,
\end{equation*}
By Lemma \ref{paris}, since $f(x) = xe^{1-x}$ is an increasing function on $[0,1]$, it is easy to see that
\begin{equation*}
\sup_{z<1-\frac{1}{2R}\epsilon}z^{1-d/2}\frac{F_{\nu_l,k}(\nu_lz)}{F_{\nu_l,k}(\nu_l)}
\leq e^{(1-z)(d/2-1)}[ze^{(1-z)}]^{\nu_l+1-d/2} \lesssim q^{\nu_l+1-d/2},
\end{equation*}
where
\begin{equation*}
q = \left(1-\frac{1}{2R}\epsilon\right)e^{\frac{1}{2R}\epsilon} < 1.
\end{equation*}
This implies the desired result.
\end{proof}

\begin{remark}[\textbf{whispering gallery modes and critical modes}]
The above corollary indicates that the eigenfunctions are localized around the outer boundary of the spherical shell as $l\rightarrow\infty$ and $k$ is fixed, leading to whispering gallery modes, as depicted in Fig. \ref{modes}(a). Moreover, in the critical case of $w = s(R)$, the eigenfunctions are localized around the inner boundary of the spherical shell. These eigenfunctions will be referred to as critical modes, as depicted in Fig. \ref{modes}(b).
\end{remark}

\begin{remark}[\textbf{Breaking of focusing modes}]
Interestingly, the eigenfunctions are not localized in a spherical shell as $k\rightarrow\infty$ and $l$ is fixed. Fig. \ref{modes}(c),(d) illustrate two different shapes of eigenfunctions under this limit, where no localization could be observed. This breaking of focusing modes, together with the occurrence of whispering gallery modes, can be viewed as limiting cases of our two-parameter localization. As $l\rightarrow\infty$ and $k$ is fixed, we have $l/k\rightarrow\infty$. In the limiting case of $w\rightarrow\infty>s(R)$, the localized radius $w/g_R(w)\rightarrow R$ and thus the eigenfunctions are localized around the outer boundary of the spherical shell, giving rise to whispering gallery modes. On the other hand, as $k\rightarrow\infty$ and $l$ is fixed, we have $l/k\rightarrow 0$. The limiting case of $w\rightarrow 0<s(R)$ belongs to the subcritical case and thus fails to produce localization.
\end{remark}

\begin{figure}
\centerline{\includegraphics[width=0.7\textwidth]{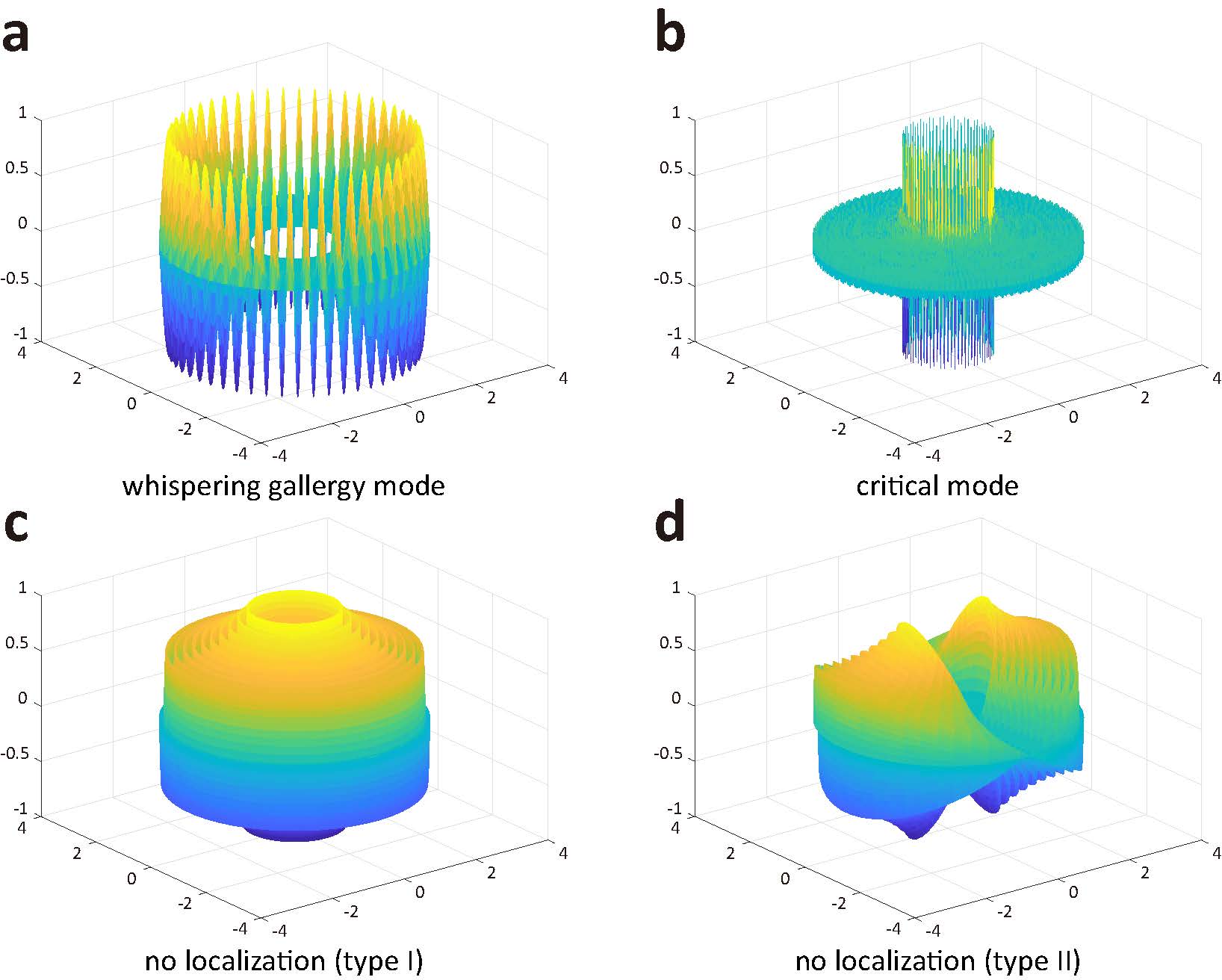}}
\caption{\textbf{Whispering gallery modes, critical modes, and the breaking of focusing modes for the Schr\"{o}dinger operator in spherical shells.} (a) Whispering gallery modes, where we take $l = 50$ and $k = 1$. (b) Critical modes, where we take $l = 197500$ and $k = 10000$. (c),(d) Breaking of focusing modes. (c) $l = 0$ and $k = 20$. (d) $l = 1$ and $k = 20$. The parameter $c = 1$ in all cases and the eigenfunctions are normalized so that the supreme norm is $1$.}\label{modes}
\end{figure}

\section{Localization for Schr\"{o}dinger eigenfunctions in sectors}
In the section, we consider the eigenvalue problem \eqref{eigenproblem} in high-dimensional sectors and annulus sectors. We first consider the two-dimensional case when the domain $\Omega$ is a sector with polar coordinate representation
\begin{equation*}
\Omega = \{(r,\theta):r<1,\;0<\theta<\beta\pi\},
\end{equation*}
where $0<\beta<2$. Similarly, all basis eigenfunctions of the eigenvalue problem \eqref{eigenproblem} are given by
\begin{equation*}
u_{kl}(r,\xi)
= J_{\nu_l}(j_{\nu_l,k}r)\sin\left(\frac{l\theta}{\beta}\right),
\end{equation*}
where $\nu_l=\sqrt{l^2+c^2}$ and $j_{\nu_l,k}$ is the $k$th zero of the Bessel function $J_{\nu_l}$. The localization behavior of the eigenfunctions in sectors is almost the same as that in balls. As $k,l\rightarrow\infty$ while keeping $l/k\rightarrow w>0$, the eigenfunctions are localized around the circular arc with polar coordinate representation
\begin{equation*}
\{(r,\theta):r = 1/h(w),\;0<\theta<\beta\pi\},
\end{equation*}
where $h(w)>1$ is defined in \eqref{inverse}. Moreover, whispering gallery modes appear as $l\rightarrow\infty$ and $k$ is fixed and focusing modes appear as $k\rightarrow\infty$ and $l$ is fixed.

We next consider the two-dimensional case when the domain $\Omega$ is an annulus sector with polar coordinate representation
\begin{equation*}
\Omega = \{(r,\theta):1<r<R,\;0<\theta<\beta\pi\},
\end{equation*}
where $R>1$ and $0<\beta<2$. Similarly, all basis eigenfunctions of the eigenvalue problem \eqref{eigenproblem} are given by
\begin{equation*}
u_{kl}(r,\xi)
= [J_{\nu_l}(a_{\nu_l,k})Y_{\nu_l}(a_{\nu_l,k}r)-Y_{\nu_l}(a_{\nu_l,k})J_{\nu_l}(a_{\nu_l,k}r)]
\sin\left(\frac{l\theta}{\beta}\right),
\end{equation*}
where $a_{\nu_l,k}$ is the $k$th zero of the cross product \eqref{crossproduct}. The localization behavior of the eigenfunctions in annulus sectors is almost the same as that in spherical shells. As $k,l\rightarrow\infty$ while keeping $l/k\rightarrow w>0$, the eigenfunctions undergo a phase transition as the $l$-$k$ ratio $w$ crosses the critical value $s(R)$. In the supercritical case of $w>s(R)$, the eigenfunctions are localized around the circular arc with polar coordinate representation
\begin{equation*}
\{(r,\theta):r = w/g_R(w),\;0<\theta<\beta\pi\},
\end{equation*}
where $g_R(w)$ is defined in \eqref{initial}. In the critical case of $w=s(R)$, the eigenfunctions are localized around the circular arc with polar coordinate representation
\begin{equation*}
\{(r,\theta):r = 1,\;0<\theta<\beta\pi\}.
\end{equation*}
In the subcritical case of $w<s(R)$, no localization could be observed. The whispering gallery modes appear as $l\rightarrow\infty$ and $k$ is fixed, while the focusing modes will not appear as $k\rightarrow\infty$ and $l$ is fixed.

In the high-dimensional case of $d\geq 3$, a sector corresponds to a subdomain of the unit sphere centered at the origin and forms a spatial angle. The spatial angle corresponds to a subdomain of the unit sphere in the same way that a planar angle corresponds to an arc of the unit circle. The conclusions in the high-dimensional case are almost the same as those in the two-dimensions case.

\section{Discussion}
In this paper, we present an detailed analysis of the two-parameter high-frequency localization for the eigenfunctions of a Schr\"{o}dinger operator with an inverse square potential in spherically symmetric domains under the Dirichlet boundary condition. The localization is realized as the azimuthal quantum number $l$ and principal quantum number $k$ tend to infinity simultaneously while keeping their ratio $l/k$ asymptotically as a constant. Four types of domains are considered: balls, sectors, spherical shells, and annulus sectors. We prove that the Schr\"{o}dinger eigenfunctions in the former two types of domains are localized around a sphere or a circular arc whose radius is increasing with respect to the $l$-$k$ ratio. In addition, we show that the decaying speeds of the eigenfunctions inside and outside the localized sphere are different: the eigenfunctions decay exponentially inside the localized sphere and decay polynomially outside. We also show that the eigenfunctions are $L^p$-localized for any $p\geq 1$ inside the localized sphere and are $L^p$-localized for any $p>4$ outside.

More interestingly, we find that the Schr\"{o}dinger eigenfunctions in the latter two types of domains experience a phase transition when the $l$-$k$ ratio crosses a critical value. In the supercritical and critical cases, the eigenfunctions are localized around a sphere or a circular arc whose radius is increasing with respect to the $l$-$k$ ratio. In the subcritical case, no localization of the eigenfunctions could be observed, leading to a new phenomenon of localization breaking.

Our results generalize the previous results about the one-parameter high-frequency localization for the eigenfunctions of the Laplacian from two different aspects. First, the Laplacian is a special case of our Schr\"{o}dinger operator when the constant $c$ vanishes. Second, the one-parameter localization can be viewed as limiting cases of our two-parameter localization when the $l$-$k$ ratio is infinity or zero, giving rise to whispering gallery modes and focusing modes in balls and only whispering gallery modes in spherical shells.

Technically, we use the $L^p$ norms to characterize two-parameter localization in the present paper. It would be interesting to study whether the Schr\"{o}dinger eigenfunctions are localized under other norms such as the $C^k$ norms, H\"{o}lder norms, and Sobolev norms. Two future challenges are expected to be solved. The first is to investigate the two-parameter localization for Schr\"{o}dinger eigenfunctions in other domains such as ellipsoids and polyhedrons. The second is to generalize the inverse square potential to other deterministic or stochastic potential functions.

\setlength{\bibsep}{5pt}
\small\bibliographystyle{nature}

\end{document}